\documentclass[12pt,a4paper]{article}
\usepackage{amsmath,amssymb,amsthm} 
\newtheorem{lemma}{Lemma}
\newtheorem{theorem}{Theorem}
\newtheorem{prop}{Proposition}
\newtheorem{cor}{Corollary}
\theoremstyle{remark}
\newtheorem{remark}{Remark}
\newcommand{\beq}{\begin{equation}}
\newcommand{\eeq}{\end{equation}}
\newcommand{\beqnn}{\begin{equation*}}
\newcommand{\eeqnn}{\end{equation*}}
\newcommand{\beqnarray}{\begin{eqnarray}}
\newcommand{\eeqnarray}{\end{eqnarray}}
\newcommand{\rd}{\partial}

\newcommand{\diag}{\operatorname{diag}}
\newcommand{\tp}[1]{\,{\vphantom{#1}}^\mathrm{t}\!\,#1}

\newcommand{\CC}{\mathbb{C}}
\newcommand{\PP}{\mathbb{P}}
\newcommand{\RR}{\mathbb{R}}
\newcommand{\ZZ}{\mathbb{Z}}
\newcommand{\bszero}{\boldsymbol{0}}
\newcommand{\bsc}{\boldsymbol{c}}
\newcommand{\bsr}{\boldsymbol{r}}
\newcommand{\bss}{\boldsymbol{s}}
\newcommand{\bst}{\boldsymbol{t}}
\newcommand{\bsu}{\boldsymbol{u}}
\newcommand{\bsx}{\boldsymbol{x}}
\newcommand{\bsy}{\boldsymbol{y}}
\newcommand{\bsrho}{\boldsymbol{\rho}}
\newcommand{\bssigma}{\boldsymbol{\sigma}}

\newcommand{\calB}{\mathcal{B}}
\newcommand{\calC}{\mathcal{C}}
\newcommand{\calD}{\mathcal{D}}

\newcommand{\calK}{\mathcal{K}}
\newcommand{\calM}{\mathcal{M}}
\newcommand{\calO}{\mathcal{O}}
\newcommand{\calP}{\mathcal{P}}
\newcommand{\calQ}{\mathcal{Q}}
\newcommand{\calR}{\mathcal{R}}
\newcommand{\calS}{\mathcal{S}}
\newcommand{\calW}{\mathcal{W}}
\newcommand{\GL}{\mathrm{GL}}
\newcommand{\SU}{\mathrm{SU}}
\newcommand{\Gr}{\mathrm{Gr}}
\newcommand{\Fr}{\mathrm{Fr}}

\begin{document}

\title{Direct linearization, Cauchy matrix and Sato Grassmannian}
\author{Kanehisa Takasaki
\thanks{E-mail: takasaki.kanehisa.65a@st.kyoto-u.ac.jp}\\
{\normalsize Osaka Central Advanced Mathematical Institute}\\ 
{\normalsize Osaka Metropolitan University}\\
{\normalsize 3-3-138 Sugimoto, Sumiyoshi-ku Osaka 558-8585, Japan}}
\date{}
\maketitle 

\begin{abstract}
Fu and Nijhoff's direct linearization scheme for 
the KP hierarchy employs a system of evolution equations 
of an infinite matrix $U$.  Part of the matrix elements 
of $U$ can be identified with affine coordinates of 
the top cell of the Sato Grassmannian.  The evolution 
equations of these matrix elements are identical to 
the KP hierarchy realized as a dynamical system on 
the Sato Grassmannian.  The system of evolution equations 
of $U$ itself can be extended by the so called negative 
flows, and the extended system turns out to be equivalent 
to the two-component KP hierarchy.  The Cauchy matrix 
approach to the KP hierarchy can be re-interpreted in 
this geometric perspective.  This geometric interpretation 
of the direct linearization scheme persists to be useful 
for a multi-component generalization.  The AKNS and 
anti-self-dual Yang-Mills hierarchies fall into this scheme.  
Multi-component Sato Grassmannians show up therein as 
the relevant geometric structure. 
\end{abstract}

\begin{flushleft}
\end{flushleft}

\newpage

\section{Introduction}

Direct linearization was proposed  
by Fokas and Ablowitz \cite{FA81,FA83} 
as a method for solving nonlinear integrable systems 
by singular integral equations  
in the spectral variable.  The idea of 
direct linearization was further developed 
by the Dutch group \cite{NLQCJ82,NQC83,NQLC83,QNCL84} 
and applied to a broad class of equations 
including discrete integrable systems. 
This method enables one to solve the system 
in question without employing 
conventional techniques such as the Lax equations, 
the inverse scattering transformations and 
the Riemann-Hilbert problem.  
This characteristic is particularly suited 
to discrete integrable systems.  
The Dutch group also introduced an infinite 
matrix structure into their method. 
This led to the next stage of progress. 

Fu and Nijhoff \cite{FN17,FN18,Fu18,FuThesis} 
fully reformulated the direct linearization scheme 
in terms of infinite matrices and applied it 
to the discrete KP equation, the KP hierarchy 
and their relatives. A central role is played 
by an infinite ($\ZZ\times\ZZ$) matrix $U$ 
with the rows and columns indexed by all integers.  
This matrix satisfies a system of evolution equations 
with quadratic nonlinearity.  This system, 
which we call \textit{the $U$-system} provisionally,
can be mapped to a linear system of evolution equations 
for another infinite matrix $C$ by a transformation 
of matrices. Conceptually, this transformation 
amounts to the direct/inverse-scattering 
transformations.  Although this transformation 
itself is somewhat formal and needs to be justified 
in each case of applications, it clarifies 
integrable nature of the evolution equations.  
Fu and Nijhoff considered reductions and 
special solutions of the KP-type systems 
within this framework.  The special solutions, 
referred to as `soliton solutions', coincide 
with the special solutions obtained 
in the Cauchy matrix approach 
\cite{NAH09,ZZ13,FZ13,XZZ14}.  
Thus the Cauchy matrix approach may be 
thought of as part of direct linearization. 

The quadratic nonlinearity of the $U$-system 
is reminiscent of a feature of the KP hierarchy. 
It is well known that the KP hierarchy 
can be identified with a dynamical system 
on an infinite dimensional Grassmann manifold 
\cite{SS82,SW85,Sato89}.  
This Grassmann manifold $\Gr$, 
commonly called \textit{the Sato Grassmannian}, 
is a set of linear subspaces 
$W \simeq \calK^{\ZZ_{<0}}$ in the fixed 
linear space $V = \calK^\ZZ$
\footnote{Throughout this paper, 
the base field $\calK$ ($\RR$, $\CC$, etc) 
is fixed and suppressed in the notation 
of $\calK$-dependent notions such as 
the Sato Grassmannian $\Gr(\calK)$, 
the general linear group $\GL(n,\calK)$, etc. 
as $\Gr$, $\GL(n)$, $\ldots$.}.  
Choosing a basis of $W$, 
one can construct a rectangular 
($\ZZ\times\ZZ_{<0}$) matrix $\xi$ 
comprising the basis vectors.  This matrix 
can be used to represent a point of $\Gr$. 
One can extract from $\xi$ an infinite number 
of affine coordinates on an open subset of $\Gr$ 
called the top cell.   
These coordinates are labelled by 
two non-negative integers as $w_{ij}$, $i,j \ge 0$.  
The KP hierarchy can be expressed as a system of 
evolution equations of $w_{ij}$'s 
\cite{Takasaki89,Takasaki89RMP}.  
Remarkably, this is a nonlinear system with 
quadratic nonlinearity.  One can see easily 
that $w_{ij}$'s correspond to a quarter 
of all matrix elements $u_{ij}$ of $U$, 
and that the evolution equations of $w_{ij}$'s 
are exactly the same as the evolution equations 
of the corresponding matrix elements of $U$. 
In other words, the KP hierarchy is literally 
a subsystem of the $U$-system. 

Our primary results comprise the following more detailed 
description of integrable and geometric structures 
in this link between Fu and Nijhoff's $U$-system 
and the Sato Grassmannian: 

\begin{itemize}
\item[1.]
Fu and Nijhoff use a vector-valued wave function 
$\bsu(z)$ alongside the matrix $U$. 
We introduce an infinite matrix $g$ that plays 
the role of dressing operator for $\bsu(z)$. 
This dressing matrix $g$ provides an alternative 
Grassmannian description of the KP hierarchy 
in the $U$-system.  
\item[2.]
We can extend the $U$-system with the so called 
\textit{negative flows}.  The negative flows yield 
another copy of the KP hierarchy in the extended 
$U$-system.  We introduce a new vector-valued 
wave function $\bar{\bsu}(z)$ and an associated 
dressing matrix $\bar{g}$ to describe 
the Grassmannian structure of the negative flows. 
\item[3.]
The pair $(g,\bar{g})$ of two dressing matrices 
carries all data of $U$, and represents a point 
of the two-component Sato Grassmannian $\Gr^{(2)}$. 
With this mapping $U \mapsto (g,\bar{g})$, 
the extended $U$-system turns out to be 
equivalent to the two-component KP hierarchy 
realized as a dynamical system on $\Gr^{(2)}$. 
\item[4.]
Fu and Nijhoff's soliton solutions can be extended
to accommodate both positive and negative flows, 
and illustrate the relationship between 
the Cauchy matrix approach and the Sato Grassmannian.  
\end{itemize}

We further extend these results 
to direct linearization of the AKNS 
and ASDYM (anti-self-dual Yang-Mills) hierarchies  
\footnote{Due to a lack of spaces, considerations 
on negative flows are omitted.}. 
These systems can be captured by the matrix $U$ 
in which the scalar matrix elements $u_{ij}$ 
are replaced by $r \times r$ matrices.  
This is a multi-component generalization 
of Fu and Nijhoff's $U$-system.  
An underlying integrable structure of this system 
is the $r$-component KP hierarchy.  
As well known in the theory of the multi-component 
KP hierarchy \cite{SS82,DJKM81,KvdL03}, 
the AKNS hierarchy can be derived 
from the two-component KP hierarchy by 
imposing a constraint. Phrased more generally, 
this amounts to cutting out a submanifold 
of the multi-component Sato Grassmannian 
preserved by the KP flows.   
The same submanifold is also used in 
the Grassmannian description of the ASDYM equation 
\cite{Takasaki83,Takasaki84CMP} 
and the associated hierarchy of higher flows 
\cite{Nakamura88,Takasaki90CMP,ACT93}. In order 
to derive these systems, we impose a constraint 
to the multi-component $U$-system. 
Such a constrained $U$-system has been studied, 
implicitly or explicitly, by Zhao and Li et al. 
in their researches on the Cauchy matrix approach 
\cite{Zhao18,LQYZ22,LQZ23} and 
the direct linearization scheme of the AKNS 
and ASDYM hierarchies \cite{LZ25}.  
We address this issue in a more systematic way 
on the basis of our ideas developed 
for Fu and Nijhoff's $U$-system.  

A novel feature in the case of the ASDYM hierarchy 
is that the evolution equations of $U$ 
take a drastically different form.  
This feature is inherited from the structure 
of the Lax formalism of the ASDYM hierarchy 
\cite{Nakamura88,Takasaki90CMP,ACT93}, 
which is also reflected to the Grassmannian description 
\cite{Takasaki83,Takasaki84CMP}. 
The structure of special solutions constructed 
by the method of Darboux transformations 
\cite{LHHZ25} is also affected by this feature.  
These facts seem to be overlooked in the preceding 
researches \cite{Zhao18,LQYZ22,LQZ23,LZ25} 
on the Cauchy matrix approach and 
the direct linearization scheme. 
The exotic feature of the ASDYM hierarchy 
stems from its relationship with 
twistor theory \cite{MWbook96}. 
We examine in detail the difference 
between the AKNS and ASDYM hierarchies 
in the perspective of direct linearization. 

This paper is organized as follows.  
In section 2, we consider Fu and Nijhoff's 
original $U$-system and its relationship 
with the KP hierarchy and the Sato Grassmannian. 
Special solutions of Fu and Nijhoff 
are also explained in some detail. 
In Section 3, we introduce negative flows 
to Fu and Nijhoff's $U$-system.  
The negative flows are shown to be related 
to another copy of the KP hierarchy and 
the Sato Grassmannian.  The relationship 
with the two-component KP hierarchy 
is explained in the language of wave functions 
and dressing matrices.  The special solutions 
of Fu and Nijhoff are also examined in more detail. 
Sections 4 and 5 are devoted to the AKNS and ASDYM 
hierarchies.  In Section 4, we introduce the notion of 
the constrained multi-component $U$-matrix.  
In section 5, we formulate the $U$-systems 
for the AKNS and ASDYM hierarchies and examine 
their implications.  Several appendices 
are added to provide relevant information 
scattered in the literature.

\section{Direct linearization of KP hierarchy}

\subsection{Evolution equations of $U$-matrix}

Let $\bst = (t_1,t_2,\ldots)$ be the set of 
the independent variables of the KP hierarchy. 
These variables are commonly called 
\textit{time variables}, but the first one 
$t_1$ is identified with the spatial coordinate $x$ 
with which we construct pseudo-differential 
operators for the Lax formalism (see Appendix A). 

Fu and Nijhoff's approach to direct linearization 
of the KP hierarchy \cite{FN17,FN18,FuThesis} 
is based on the evolution equations 
\beq
  \frac{\rd U}{\rd t_k} 
  = \Lambda^k U - U\Lambda^{-k} - U\calO_k U, 
  \quad k = 1,2,\ldots, 
  \label{2:U-tk-eq}
\eeq
of the $\bst$-dependent $\ZZ\times\ZZ$ matrix 
$U = (u_{ij})_{i,j\in\ZZ}$.  
We call these equations the $U$-system. 
The scalar-valued matrix elements $u_{ij}$ 
are the main dynamical variables 
of Fu and Nijhoff's infinite matrix formalism.  

$\Lambda^k$ and $\calO_k$ in (\ref{2:U-tk-eq}) 
are constant $\ZZ\times\ZZ$-matrices of the form 
\[
  \Lambda^k = \sum_{i\in\ZZ}E_{i,i+k},~~
  \calO_k = \sum_{i=0}^{k-1} E_{i,k-i-1}, 
\]
where $E_{ij}$ denotes the matrix unit in which 
all matrix elements vanish except for 
the $(i,j)$-th element being equal to $1$. 
Whereas $\Lambda^k$'s are the so called 
\textit{shift matrices}, $\calO_k$'s 
are obtained from the matrix 
\[
  \Omega = \sum_{i=0}^\infty E_{i,-i-1} 
\]
as 
\beq
  \calO_k = \Omega\Lambda^k - \Lambda^{-k}\Omega. 
  \label{2:OmLam-rel}
\eeq
$\Omega$ is related to the Cauchy kernel 
\[
  \Omega(z,w) = \frac{1}{z - w} 
\]
in the sense that 
\beq
  \sum_{i,j\in\ZZ}\Omega_{ij}w^iz^j 
  = \sum_{i=0}^\infty\frac{w^i}{z^{i+1}} 
  = \frac{1}{z - w} 
  \quad \text{for $|z| > |w|$}. 
  \label{2:Om-CK}
\eeq

The compatibility of the flows of (\ref{2:U-tk-eq}) 
is a consequence of the algebraic relations 
(\ref{2:OmLam-rel}) as shown below. 

\begin{prop}
\beq
  \frac{\rd}{\rd t_l}\frac{\rd U}{\rd t_k} 
  = \frac{\rd}{\rd t_k}\frac{\rd U}{\rd t_l}.  
  \label{2:tktl-comm}
\eeq
\end{prop}

\proof
Using (\ref{2:U-tk-eq}), we can calculate 
the left hand side of (\ref{2:tktl-comm}) as 
\[
\begin{aligned}
  \frac{\rd}{\rd t_l}\frac{\rd U}{\rd t_k} 
  &= \frac{\rd}{\rd t_l}
     (\Lambda^kU - U\Lambda^{-k} - U\calO_kU) \\
  &= \Lambda^k\frac{\rd U}{\rd t_l} 
     - \frac{\rd U}{\rd t_l}\Lambda^{-k}
     - \frac{\rd U}{\rd t_l}\calO_kU 
     - U\calO_k\frac{\rd U}{\rd t_l}\\
  &= \Lambda^k(\Lambda^lU - U\Lambda^{-l} - U\calO_lU) 
     - (\Lambda^lU - U\Lambda^{-l} - U\calO_lU)\Lambda^{-k} \\
  &~~\mbox{}
     - (\Lambda^lU - U\Lambda^{-l} - U\calO_lU)\calO_kU 
     - U\calO_k(\Lambda^lU - U\Lambda^{-l} - U\calO_lU). 
\end{aligned}
\]
We obtain the right hand side (\ref{2:tktl-comm}) 
by exchanging $k$ and $l$.  Therefore the difference 
becomes 
\[
  \frac{\rd}{\rd t_l}\frac{\rd U}{\rd t_k} 
  - \frac{\rd}{\rd t_k}\frac{\rd U}{\rd t_l} 
  = U(\Lambda^{-l}\calO_k - \calO_k\Lambda^l 
       - \Lambda^{-k}\calO_l + \calO_l\Lambda^k)U. 
\]
The algebraic relations (\ref{2:OmLam-rel}) imply that 
the right hand side of the last equation vanishes. 
\qed

The evolution equations (\ref{2:U-tk-eq}) of $U$ 
are connected with the linear evolution equations 
\beq
  \frac{\rd C}{\rd t_k} = \Lambda^k C - C\Lambda^{-k}, 
  ~~ k = 1,2,\ldots 
  \label{2:C-tk-eq}
\eeq
of another $\ZZ\times\ZZ$ matrix 
$C = (c_{ij})_{i,j\in\ZZ}$ via the algebraic relation 
\beq
  U = (E - U\Omega)C, 
  \label{2:UC-rel}
\eeq
where $E$ is the $\ZZ\times\ZZ$ unit matrix. 
One can solve this relation for $U$ as 
\[
  U = C(E + \Omega C)^{-1}
\]
as far as $E + \Omega C$ is invertible. 
This is a transformation that sends solutions 
of (\ref{2:C-tk-eq}) to those of (\ref{2:U-tk-eq}).  
The inverse transformation 
\[
  C = (E - U\Omega)^{-1}U 
\]
sends solutions of (\ref{2:U-tk-eq}) 
to solution of (\ref{2:C-tk-eq}).  
The quadratic term $U\calO_kU$ of 
(\ref{2:U-tk-eq}) is thus eliminated 
by these transformations, and we are left 
with the \textit{linearized equations} (\ref{2:C-tk-eq}).  
This is also another explanation of the compatibility 
of the flows of (\ref{2:U-tk-eq}), because 
the linearized equations are obviously compatible. 

These transformations may be thought of 
as an analogue of the direct/inverse scattering 
transformations.  The matrix $C$ amounts 
to the \textit{scattering data} and obeys 
the simple rule of time evolutions 
\[
  C = \rho(\Lambda)C(\bszero)\sigma(\Lambda^{-1}), 
  ~~ C(\bszero) = C|_{\bst=\bszero}, 
\]
where 
\[
  \rho(z) = \exp\left(\sum_{k=1}^\infty t_kz^k\right),~~
  \sigma(z) = \exp\left(- \sum_{k=1}^\infty t_kz^k\right). 
\]

\begin{remark}
The evolution equations (\ref{2:U-tk-eq}) of $U$ 
are slightly different from the equations 
employed by Fu and Nijhoff \cite{FN17,FN18,FuThesis}. 
In their formulation, the second term 
of the right hand side has an extra sign factor $(-1)^k$.  
This is a superficial difference; the sign factor 
can be eliminated by the simple transformation 
$U \to U\diag((-1)^i)$. This, however, affects 
other equations and the definition of 
$\Omega$ and $\calO_k$.  The necessary modification 
has been done in our definition of 
$\Omega$ and $\calO_k$. With this modification, 
we can avoid some cumbersome sign factors.
\end{remark}

\subsection{First contact with Sato Grassmannian}

The relationship with the Sato Grassmannian \cite{SS82} 
can be readily seen from the evolution equations 
\beq
  \frac{\rd u_{ij}}{\rd t_k} 
  = u_{i+k,j} - u_{i,j+k} 
    - \sum_{l=0}^{k-1}u_{il}u_{k-l-1,j} 
  \label{2:uij-tk-eq}
\eeq
obtained from (\ref{2:U-tk-eq}) 
for the matrix elements of $U$.  
The quadratic terms on the right hand side 
originate in the term $U\calO_kU$ 
of (\ref{2:C-tk-eq}).  Among these equations, 
those in the quarter $i,j \ge 0$ form 
a closed subsystem.  This subsystem turns out 
to be identical to the differential equations 
\beq
  \frac{\rd w_{ij}}{\rd t_k} 
  = w_{i+k,j} - w_{i,j+k} 
    - \sum_{l=0}^{k-1}w_{il}w_{k-l-1,j}, 
  ~~ i,j \ge 0, 
  \label{2:wij-tk-eq}
\eeq
that describe the KP hierarchy 
as a dynamical system on the Sato Grassmannnian 
\cite{Takasaki89,Takasaki89RMP} (see Appendix A). 
$w_{ij}$'s, $i,j \ge 0$, are affine coordinates 
of the top cell therein.  
We are thus led to the identification 
\beq
  u_{ij} = w_{ij}~~ \text{for $i,j \ge 0$}. 
  \label{2:u-w-rel}
\eeq

The affine coordinates $w_{ij}$ are 
building blocks of the following 
$\ZZ\times\ZZ_{<0}$ and $\ZZ_{\ge 0}\times\ZZ$ matrices 
\footnote{$\ZZ_{\ge 0}$ and $\ZZ_{<0}$ 
denote the sets of all non-negative 
and negative integers, respectively.}:  
\[
\begin{aligned}
  \xi &= \left(\begin{array}{c}
         \delta_{ij}\\\hline 
         w_{i,-j-1}
         \end{array}\right)_{i\in\ZZ,j\in\ZZ_{<0}}
  = \begin{pmatrix}
    \ddots&\vdots &\vdots\\
    \cdots& 1 & 0\\
    \cdots& 0 & 1\\
    \cdots& w_{01} & w_{00}\\
    \cdots& w_{11} & w_{10}\\
          &\vdots &\vdots
    \end{pmatrix},\\
  \eta &= \left(\begin{array}{c|c}
         - w_{i,-j-1} & \delta_{ij}
         \end{array}\right)_{i\in\ZZ_{\ge 0},j\in\ZZ}
  = \begin{pmatrix}
    \cdots &-w_{01} &-w_{00} & 1 & 0 & \cdots\\
    \cdots &-w_{11} &-w_{10} & 0 & 1 & \cdots\\
           &\vdots &\vdots &\vdots &\vdots&\ddots
    \end{pmatrix}. 
\end{aligned}
\]
$\xi$ is divided into the upper ($i,j < 0$) 
and lower ($i \ge 0,j < 0$) blocks, 
and $\eta$ into the left ($i \ge 0, j< 0$) 
and right ($i,j \ge 0$) blocks. 
By construction, $\xi$ and $\eta$ satisfy 
the orthogonality condition 
\beq
  \eta\xi = 0
  \label{2:eta-xi-rel}
\eeq
and represent a common point of the Sato Grassmannian 
\[
  \Gr \simeq \Fr(\ZZ,\ZZ_{<0})/\GL(\ZZ_{<0}) 
      \simeq \GL(\ZZ_{\ge 0})\backslash\Fr(\ZZ_{\ge 0},\ZZ) 
\]
realized as a coset space.  $\Fr(I,J)$ denotes 
the set of $I \times J$-matrices of \textit{full rank}, 
and $\GL(I)$ and $\GL(J)$, respectively, 
the general linear group of $I\times I$ 
and $J\times J$ matrices (all in a loose sense
\footnote{These notions are defined rigorously 
in the work of M.~Sato and Y.~Sato \cite{SS82}.}).
In the first expression 
$\Gr \simeq \Fr(\ZZ,\ZZ_{<0})/\GL(\ZZ_{<0})$, 
$\xi$ and $\xi h$, $h \in \GL(\ZZ_{<0})$, 
represent the same point of $\Gr$ which 
is understood to be a linear space 
spanned by the column vectors of $\xi$.  
In the second expression 
$\Gr \simeq \GL(\ZZ_{\ge 0})\backslash\Fr(\ZZ_{\ge 0},\ZZ)$, 
the roles of rows and columns are exchanged. 

The equations (\ref{2:wij-tk-eq}) for $u_{ij}$'s 
can be cast into the two equivalent matrix forms
\beq
  \frac{\rd\xi}{\rd t_k} = \Lambda^k\xi - \xi\calB_k, 
  \label{2:xi-tk-eq}
\eeq
\beq
  \frac{\rd\eta}{\rd t_k} = \calC_k\eta - \eta\Lambda^k, 
  \label{2:eta-tk-eq}
\eeq
where 
\[
\begin{aligned}
  \calB_k &= \left(\begin{array}{cl}
         \delta_{i+k,j} & (i<-k)\\\hline
         w_{i+k,-j-1} & (-k\le i<0)
         \end{array}\right)_{i,j\in\ZZ_{<0}},\\
  \calC_k &= \left(\begin{array}{c|c}
         - w_{i,k-j-1}~~(0\le j<k) & \delta_{i,j-k}~~(j\ge k)
         \end{array}\right)_{i,j\in\ZZ_{\ge 0}}. 
\end{aligned}
\]
In fact, these equations are in a special 
\textit{gauge}. In a general gauge, 
they allow gauge transformations 
\[
\begin{aligned}
  &\xi \to \xi h,~~
  B_k \to h^{-1}B_kh - h^{-1}\frac{\rd h}{\rd t_k},\\
  &\eta \to \tilde{h}\eta,~~
  C_k \to \tilde{h}C_k\tilde{h}^{-1}  
   + \frac{\rd\tilde{h}}{\rd t_k}\tilde{h}^{-1} 
\end{aligned}
\]
with arbitrary $\bst$-dependent elements 
$h = h(\bst)$ of $\GL(\ZZ_{<0})$ and 
$\tilde{h} = \tilde{h}(\bst)$ of $\GL(\ZZ_{\ge 0})$.  
Thus (\ref{2:xi-tk-eq}) and (\ref{2:eta-tk-eq}) 
define a geometrically consistent dynamical system 
on the Sato Grassmannian.  
This interpretation of (\ref{2:xi-tk-eq}) and 
(\ref{2:eta-tk-eq}) lies in the heart 
of the geometric approach to the KP hierarchy 
\cite{SS82,Sato89,Takasaki89RMP}.  

We are thus led to the following conclusion: 

\begin{prop}
The $U$-system (\ref{2:U-tk-eq}) contains 
a closed subsystem for the matrix elements 
$u_{ij}$, $i,j \ge 0$, of the upper-left part of $U$. 
These matrix elements can be identified with 
affine coordinates of the top cell 
of the Sato Grassmannian $\Gr$.  The $U$-system 
restricted to these matrix elements 
is equivalent to the KP hierarchy realized 
as dynamical system on $\Gr$. 
\end{prop}

This result raises the question of 
what geometric meaning the other part of $U$ holds.  
In order to address this question, we now turn to 
another important ingredient of Fu and Nijhoff's 
direct linearization scheme, namely,  
\textit{the vector-valued wave function}.

\subsection{Vector-valued wave function and 
Sato Grassmannian}

The vector-valued wave function 
$\bsu(z) = (u_i(z))_{i\in\ZZ}$ is defined by $U$ as 
\beq
  \bsu(z) = (E - U\Omega)\bsc(z)\rho(z), 
  ~~\bsc(z) = \left(z^i\right)_{i\in\ZZ}. 
  \label{2:u(z)-def}
\eeq
The components $u_i(z)$ can be written as 
\beq
  u_i(z) 
  = \left(z^i - \sum_{j=0}^\infty u_{ij}z^{-j-1}\right)\rho(z). 
  \label{2:u_i(z)-uij}
\eeq
$\bsu(z)$ satisfies the auxiliary linear equations 
\beq
  \frac{\rd\bsu(z)}{\rd t_k} 
  = (\Lambda^k - U\calO_k)\bsu(z), ~~ k = 1,2,\ldots. 
  \label{2:u(z)-tk-eq}
\eeq
In terms of the components $u_i(z)$, these equations 
can be expressed as 
\beq
  \frac{\rd u_i(z)}{\rd t_k}
  = u_{i+k}(z) - \sum_{l=0}^{k-1}u_{il}u_{k-l-1}(z), 
  ~~ i \in \ZZ. 
  \label{2:ui(z)-tk-eq}
\eeq
These wave functions originate in the conventional 
formulation of direct linearization 
employing singular integral equations 
\cite{FA81,FA83,NLQCJ82,NQC83,NQLC83,QNCL84}.  
Fu and Nijhoff use the vector-valued 
wave function $\bsu(z)$ very efficiently 
to construct special solutions as well 
\cite{FN17,FN18,FuThesis}.  
We now rephrase the aforementioned 
Grassmannian structure in terms of $\bsu(z)$. 

Let us consider the matrix 
\beq
  g = E- U\Omega 
  \label{2:g-def} 
\eeq
that generates $\bsu(z)$ from the `plane wave' 
$\bsc(z)\rho(z)$. 
Divided into the upper-left ($i,j<0$), 
lower-left ($i\ge 0,j<0$), 
upper-right ($i<0,j\ge 0$) and 
lower-right ($i,j\ge 0$) blocks, 
this matrix can be expressed as 
\beq
  g = \left(\begin{array}{c|c}
    \delta_{ij} - u_{i,-j-1} & 0 \\\hline
    - u_{i,-j-1} & \delta_{ij} 
    \end{array}\right). 
  \label{2:g-blocks}
\eeq
The central part of this matrix looks as follows: 
\[
  g = \begin{pmatrix}
      \ddots &\vdots &\vdots &\vdots & \vdots \\
      \cdots &1-u_{-2,1} & -u_{-2,0} & 0 & 0 & \cdots\\
      \cdots &-u_{-1,1} &1-u_{-1,0} & 0 & 0 & \cdots\\
      \cdots &-u_{0,1} & -u_{0,0} & 1 & 0 & \cdots\\
      \cdots &-u_{1,1} & -u_{1,0} & 0 & 1 & \cdots\\
             &\vdots&\vdots &\vdots &\vdots &\ddots
      \end{pmatrix}. 
\]
Note that a half of all matrix elements of $U$, 
i.e. $u_{ij}$ for $i \in \ZZ$ and $j \ge 0$, 
are included therein.  
Also note that the lower half of $g$ is nothing 
but the foregoing matrix $\eta$, namely,   
\[
  \eta = \left(\begin{array}{c|c}
    - u_{i,-j-1} & \delta_{ij} 
    \end{array}\right)_{i\in\ZZ_{\ge 0},j\in\ZZ}. 
\]
Moreover, $g$ turns out to satisfy the following 
differential equations.  

\begin{prop}
\beq
  \frac{\rd g}{\rd t_k} 
  = (\Lambda^k - U\calO_k)g - g\Lambda^k, 
  ~~ k = 1,2,\ldots. 
  \label{2:g-tk-eq}
\eeq
\end{prop}

\proof
Differentiating $g$ directly 
and using (\ref{2:OmLam-rel}), we have 
\[
\begin{aligned}
  \frac{\rd g}{\rd t_k} 
  &= - \frac{\rd U}{\rd t_k}\Omega \\
  &= - (\Lambda^kU - U\Lambda^{-k} - U\calO_kU)\Omega \\
  &= - \Lambda^k(E - g) + U\Lambda^{-k}\Omega 
     + U\calO_k(E - g)\\
  &= (\Lambda^k - U\calO_k)g - \Lambda^k 
     + U(\Lambda^{-k}\Omega + \calO_k) \\
  &= (\Lambda^k - U\calO_k)g - \Lambda^k + U\Omega\Lambda^k 
     ~~(\text{by (\ref{2:OmLam-rel})})\\
  &= (\Lambda^k - U\calO_k)g - g\Lambda^k.  
\end{aligned}
\]
\qed

The equations (\ref{2:g-tk-eq}) for $g$ are 
very suggestive from several points of view.  
Firstly, the auxiliary linear equations 
(\ref{2:u(z)-tk-eq}) can be readily derived 
from these equations.  In this sense, 
the matrix $g$ plays the role of 
\textit{a dressing operator} that transforms 
the undressed wave function $\bsc(z)\rho(z)$ 
to the dressed wave function $\bsu(z)$. 
Secondly, these equations resemble 
the equations (\ref{2:eta-tk-eq}) 
satisfied by $\eta$.  This hints at 
a possible geometric interpretation 
of (\ref{2:g-tk-eq}) in the language of the Sato 
Grassmannian.  To this end, let us point out 
the following structural characteristic of 
the coefficient matrix $\Lambda^k - U\calO_k$ 
in (\ref{2:g-tk-eq}). 

\begin{prop}
Let $\Lambda^k - U\calO_k$ be decomposed 
into four blocks like the decomposition 
(\ref{2:g-blocks}) of $g$.  
Then the lower-left block 
vanishes and the lower-right block is equal 
to the coefficient matrix $\calC_k$ 
of (\ref{2:eta-tk-eq}), namely, 
\beq
  \Lambda^k - U\calO_k 
  = \left(\begin{array}{c|c}
    * & * \\\hline 
    0 & \calC_k 
    \end{array}\right). 
  \label{2:Lk-UOk-block}
\eeq
\end{prop}

\proof 
$\Lambda^k$ and $U\calO_k$ have 
the following block structures: 
\[
\begin{aligned}
  \Lambda^k 
  &= \left(\begin{array}{c|c|c}
    * & * & 0\\\hline 
    0 & 0~(0\le j<k) & \delta_{i,j-k}~(j\ge k) 
    \end{array}\right),\\
  U\calO_k 
  &= \left(\begin{array}{c|c|c}
    0 & * & *\\\hline
    0 & u_{i,k-j-1}~(0\le j<k) & 0~(j\ge k)
    \end{array}\right)
\end{aligned}
\]
Thus $\Lambda^k - U\calO_k$ turns out to take 
the form as shown in (\ref{2:Lk-UOk-block}). 
\qed 

This block-triangular form of 
$\Lambda^k - U\calO_k$ and the block 
structure (\ref{2:g-blocks}) of $g$ imply 
that the differential equations (\ref{2:g-tk-eq}) 
of $g$ are an extension of the differential 
equations (\ref{2:eta-tk-eq}) of $\eta$.  
In other words, (\ref{2:eta-tk-eq}) is a subsystem 
of (\ref{2:g-tk-eq}).  This is a place 
where another coset space description 
of the Sato Grassmannian 
\[
  \Gr \simeq \calP\backslash\GL(\ZZ) 
\]
comes into the game.  The stabilizer $\calP$ 
is the maximal parabolic subgroup of $\GL(\infty)$ 
(in a loose sense) that consists of invertible matrices 
of the form 
\[
    p = \left(\begin{array}{c|c}
        * & * \\\hline
        0 & * 
        \end{array}\right), 
\]
namely, all matrix elements of the lower-left block 
vanish.  The block structure (\ref{2:Lk-UOk-block}) 
of $\Lambda^k - U\calO_k$ implies that this matrix 
takes values in the Lie algebra of $\calP$.  
Thus the matrix $g$ represents a point of $\Gr$. 
This is the same point of $\Gr$ represented by $\eta$ 
in the previous coset space realization 
$\Gr \simeq \GL(\ZZ_{\ge 0})\backslash\Fr(\ZZ_{\ge 0},\ZZ)$. 

The equations (\ref{2:g-tk-eq}) for $g$ thus 
define the same dynamical system on $\Gr$ 
as defined by the equations (\ref{2:xi-tk-eq}) 
and (\ref{2:eta-tk-eq}) for $\xi,\eta$.  
To be more precise, this is a picture in a special gauge.  
The picture in a general gauge is achieved 
by replacing the special coefficient matrix 
$\Lambda^k - U\calO_k$ by a general matrix $\Gamma_k$ 
of the block-triangular form 
\[
  \Gamma_k 
   = \left(\begin{array}{c|c}
     * & * \\\hline
     0 & * 
     \end{array}\right). 
\]
The outcome are the equations 
\beq
  \frac{\rd g}{\rd t_k} = \Gamma_k g - g\Lambda^k, 
  ~~ k = 1,2,\ldots, 
  \label{2:g-tk-eq2}
\eeq
that allow gauge transformations 
\[
  g \to pg,~~
  \Gamma_k \to p\Gamma_kp^{-1} + \frac{\rd p}{\rd t_k}p^{-1}
\]
by a $\bst$-dependent element 
$p = p(\bst)$ of $\calP$.  

We are thus led to the following conclusion. 

\begin{prop}
The equations (\ref{2:g-tk-eq}) for the matrix $g$ 
define the same dynamical system 
as defined by the equations (\ref{2:eta-tk-eq}) 
for the matrix $\eta$ on the Grassmannian $\Gr$.  
\end{prop}

This result shows, in particular, 
a geometric interpretation of the matrix elements 
$u_{ij}$ of $U$ for $i\in\ZZ$ and $j \ge 0$. 
We shall consider the status of the other 
matrix elements in the next section.  

\begin{remark}
The $0$-th wave function $u_0(z)$ can be 
identified with the fundamental wave function 
\[
  \Psi(z) =\left(1 + \sum_{n=1}^\infty w_nz^{-n}\right)\rho(z),~~
  w_n = - u_{0,n-1}, 
\]
of the KP hierarchy that satisfies 
the auxiliary linear equations 
\beq
  \frac{\rd\Psi(z)}{\rd t_k} = B_k\Psi(z),~~
  B_k = \rd_x^k + b_{k2}\rd_x^{k-2} + \cdots + b_{kk},~~
  k = 1,2,\ldots 
  \label{2:Psi-tk-eq}
\eeq
(see Appendix A). 
This fact can be deduced 
from the particular ($k = 1$) case 
\[
  u_{i+1}(z) = \frac{\rd u_i(z)}{\rd t_1} + u_{i0}u_0(z) 
\]
of the differential equations (\ref{2:ui(z)-tk-eq}) 
for $u_i(z)$'s.  One can use these equations 
repeatedly to express $u_i(z)$ for $i \ge 1$ 
in terms of $u_0(z)$ as 
\beq
  u_i(z) = P_iu_0(z), 
  \label{2:ui(z)u0(z)-rel}
\eeq
where $P_i$ is a differential operator of the form 
\[
  P_i = \rd_x^i + p_{i1}\rd_x^{i-1} + \cdots + p_{ii}, 
  ~~ \rd_x = \rd/\rd t_1.  
\]
Substituting these expressions for $u_i(z)$ 
in (\ref{2:ui(z)-tk-eq}) yields linear 
differential equations of the form (\ref{2:Psi-tk-eq}) 
for $u_0(z)$. 
\end{remark}

\subsection{Special solutions and Cauchy matrix}

Fu and Nijhoff \cite{FN17,FN18,FuThesis} 
obtained a special solution of (\ref{2:U-tk-eq}) 
by setting $C$ in a special form
\footnote{Actually, they used an equivalent 
expression in terms of integral equations.}. 
This solution, called 
the \textit{soliton solution} therein, 
depends on several constants 
$\kappa_i,\lambda_j,a_{ij}$ 
($1 \le i \le N$, $1 \le j \le M$).  
$\kappa_i$'s and $\lambda_j$'s are nonzero 
and assumed to be mutually distinct, 
namely, $\kappa_i \not= \kappa_j$ 
and $\lambda_i \not= \lambda_j$  
if $i \not= j$, and $\kappa_i \not= \lambda_j$ 
for all $i,j$.  
The matrix elements of $U$ take 
the very suggestive form 
\beq
  u_{ij} 
  = \tp{\bsr}K^i(E_N + A\calM)^{-1}AL^j\bss 
  = \tp{\bsr}K^iA(E_M + \calM A)^{-1}L^j\bss, 
  \label{2:CM-uij}
\eeq
where $E_N$ and $E_M$ denote  
the unit matrix of sizes $N\times N$ and $M\times M$. 
$A$ and $\calM$ are $N\times M$ 
and $M\times N$ matrices of the form 
\[
  A = (a_{ij})_{1\le i\le N,1\le j\le M},~~
  \calM = \left(\frac{\sigma(\lambda_i)\rho(\kappa_j)}
          {\kappa_j - \lambda_i}\right)_{1\le i\le M,1\le j\le N},  
\]
$K$ and $L$ are the $N\times N$ and $M\times M$ 
diagonal matrices 
\[
  K = \diag(\kappa_1,\ldots,\kappa_N),~~
  L = \diag(\lambda_1,\ldots,\lambda_M), 
\]
and $\bsr$ and $\bss$ are the $N$- and 
$M$-dimensional vectors 
\[
  \bsr = (\rho(\kappa_i))_{i=1}^N,~~
  \bss = (\sigma(\lambda_i))_{i=1}^M.  
\]

This special solution can be derived 
from the Cauchy matrix approach 
\cite{NAH09,ZZ13,FZ13,XZZ14,FS22} as well.  
The foregoing matrix $\calM$ may be thought of 
as the \textit{Cauchy matrix} satisfying 
the \textit{Sylvester equation} 
\[
  \calM K - L \calM = \bss\tp{\bsr}. 
\]
$\bsr$ and $\bss$ satisfy the differential equations 
\[
  \frac{\rd\bsr}{\rd t_k} = K^k\bsr,~~
  \frac{\rd\bss}{\rd t_k} = - L^k\bss  
\]
of plane waves.  One can verify 
by direct calculations that $u_{ij}$'s 
defined by (\ref{2:CM-uij}) satisfy 
the differential equations (\ref{2:uij-tk-eq}). 
We shall review those calculations 
in the next section in an extended setting. 
Thus $u_{ij}$'s can be identified 
with the $S$-potentials $S^{(ij)}$ 
in the Cauchy matrix approach.  
In other words, $S^{(i,j)}$'s are 
affine coordinates of the top cell of 
the Sato Grassmannian. 

Fu and Nijhoff \cite{FN17,FN18,FuThesis} obtained 
this special solution by solving the functional equation 
\beq
  u_i(z) 
  = z^i\rho(z) 
    - \sum_{l=1}^N\sum_{m=1}^Mu_i(\kappa_l)a_{lm}
      \frac{\sigma(\lambda_m)\rho(z)}{z - \lambda_m} 
  \label{2:CM-ui(z)-eq}
\eeq
for $u_i(z)$.  This equation is a degenerate case 
of the integral equation \cite{FN17,FN18,FuThesis} 
corresponding to the matrix equation (\ref{2:UC-rel}) 
for $U$.  By substituting $z = \kappa_j$, 
(\ref{2:CM-ui(z)-eq}) yields the linear equations 
\[
  u_i(\kappa_j) 
  = \kappa_j^i\rho(\kappa_j) 
    - \sum_{l=1}^N\sum_{m=1}^Mu_i(\kappa_l)a_{lm}
      \frac{\sigma(\lambda_m)\rho(\kappa_j)}{\kappa_j - \lambda_m},
  ~~ j = 1,\ldots,N, 
\]
for $u_i(\kappa_j)$'s.  
These equations can be solved as 
\beq
  (u_i(\kappa_j))_{j=1}^N = \tp{\bsr}K^i(E_N + A\calM)^{-1}. 
  \label{2:CM-ui(pj)}
\eeq
Plugging this into (\ref{2:CM-ui(z)-eq}) 
and doing some algebra, one obtains the explicit form 
\beq
  u_i(z) 
  = \left(z^i - \tp{\bsr}K^i(E_N + A\calM)^{-1}A 
          (zE_M - L)^{-1}\bss\right)\rho(z) 
  \label{2:CM-ui(z)}
\eeq
of the solution of (\ref{2:CM-ui(z)-eq}).  

This implies, in particular, that 
the amplitude part of $u_i(z)$ 
in front of $\rho(z)$ is a rational function of $z$, 
though this fact itself is an immediate consequence 
of the functional equation (\ref{2:CM-ui(z)-eq}). 
One can extract $u_{ij}$'s for $i \in \ZZ, j \ge 0$ 
by expanding this rational function in negative powers 
of $z$ as shown in (\ref{2:u_i(z)-uij}). 
This leads to the formula (\ref{2:CM-uij}). 
We shall reconsider this special solution 
in the next section in an extended form.

\section{KP hierarchy extended by negative flows}

\subsection{KP hierarchy of negative flows} 

The differential equations (\ref{2:U-tk-eq}) of $U$ 
can be extended by the \textit{negative flows}
\beq
  \frac{\rd U}{\rd t_{-k}} 
  = \Lambda^{-k}U - U\Lambda^k - U\calO_{-k}U, 
  ~~ k = 1,2,\ldots, 
  \label{3:U-tmk-eq}
\eeq
with the new time variables 
$\bar{\bst} = (t_{-1},t_{-2},\ldots)$. 
$\calO_{-k}$'s are defined as 
\[
  \calO_{-k} = - \sum_{l=1}^kE_{-l,l-k-1}.
\]
This is a natural extrapolation of $\calO_k$'s 
as they are related with $\Omega$ as 
\beq
  \calO_{-k} = \Omega\Lambda^{-k} - \Lambda^k\Omega. 
  \label{3:OmmkLam-rel}
\eeq
Using these algebraic relations, 
one can prove the compatibility 
\[
  \frac{\rd}{\rd t_{-l}}\frac{\rd U}{\rd t_{-k}} 
  =  \frac{\rd}{\rd t_{-k}}\frac{\rd U}{\rd t_{-l}} 
\]
of the flows in exactly the same way 
as in the case of positive flows.  
Linearization by the transformation (\ref{2:UC-rel}) 
also works for the negative flows.  
The linearized equations become 
\beq
  \frac{\rd C}{\rd t_{-k}} = \Lambda^{-k}C - C\Lambda^k, 
  ~~ k = 1,2,\ldots. 
  \label{3:C-t(k)-eq}
\eeq
The compatibility of flows can be 
explained by the linearizability as well. 

Let us examine the equations (\ref{3:U-tmk-eq}) 
in more detail. These equation become 
the differential equations 
\beq
  \frac{\rd u_{ij}}{\rd t_{-k}} 
  = u_{i-k,j} - u_{i,j-k} 
    + \sum_{l=1}^k u_{i,-l}u_{l-k-1,j} 
  \label{3:uij-tmk-eq}
\eeq
for the matrix elements $u_{ij}$. 
Recall that the matrix elements $u_{ij}$, $i,j \ge 0$, 
of the lower-left block of $U$ satisfy 
a closed subsystem of the positive flows 
(\ref{2:uij-tk-eq}).  This subsystem 
is identical to the equations (\ref{2:wij-tk-eq}) 
for the affine coordinates $w_{ij}$ 
of the Sato Grassmannian $\Gr$. 
As regards the equations (\ref{3:uij-tmk-eq}) 
of the negative flows, we can see that 
the equations are closed for 
the matrix elements $u_{ij}$, $i,j < 0$, 
of the upper-left block of $U$.  
Letting 
\beq
  \bar{w}_{ij} = - u_{-i-1,-j-1} 
  ~~\text{of $i,j\ge 0$}, 
  \label{3:u-wbar-rel}
\eeq
we obtain evolution equations of the same form 
as the evolution equations (\ref{2:wij-tk-eq}) 
for $\bar{w}_{ij}$'s except that the role of 
the time variables is now played by 
the negative times $t_{-k}$.  

The evolution equations (\ref{3:U-tmk-eq}) 
of the negative flows thus turns out 
to contain another copy of the KP hierarchy. 
One can introduce analogues of $\xi$ and $\eta$ 
to translate these equations 
into the language of a Grassmannian. 
As we shall see later, it is convenient 
to choose the $\eta$-matrix to be 
a $\ZZ_{<0}\times\ZZ$ matrix of the form 
\[
  \bar{\eta} 
  = \left(\begin{array}{c|c}
    \delta_{ij} & u_{i,-j-1}
    \end{array}\right)_{i\in\ZZ_{<0},j\in\ZZ}. 
\]
The evolution equations for $\bar{w}_{ij}$'s 
of (\ref{3:u-wbar-rel}) can be thereby 
converted to the differential equations 
\beq
  \frac{\rd\bar{\eta}}{\rd t_{-k}} 
  = \bar{\calC}_k\bar{\eta} 
    - \bar{\eta}\Lambda^{-k}, 
  ~~ k = 1,2,\ldots 
  \label{3:etabar-tmk-eq}
\eeq
where 
\[
  \bar{\calC}_k 
  = \left(\begin{array}{c|c}
    \delta_{i,j+k}~(j < -k) & u_{i,-j-k-1}~(-k\le j < 0)
    \end{array}\right)_{i,j\in\ZZ_{<0}}. 
\]
These evolution equations define a dynamical system
on the complementary Grassmannian 
\[
  \overline{\Gr} 
  \simeq \GL(\ZZ_{<0})\backslash\Fr(\ZZ_{<0},\ZZ). 
\]

\subsection{Modifying Omega matrix and its implications} 

The matrix $g$ is constructed as shown 
in (\ref{2:g-def}) using the matrix $\Omega$. 
A half of all matrix elements of $U$, namely, 
$u_{ij}$ for $i\in\ZZ, j \ge 0$, are contained in $g$. 
The geometric meaning of those matrix elements 
has been elucidated in the Grassmannian perspective 
of $g$.  We now introduce another matrix $\bar{g}$ 
that explains the geometric meaning of the other half, 
namely $u_{ij}$ for $i \in \ZZ, j < 0$, in the same way.  
To this end, we need to modify the matrix $\Omega$ 
appropriately.  

The primary role of $\Omega$ is to express 
the matrices $\calO_{\pm k}$ as shown 
in (\ref{2:OmLam-rel}) and (\ref{3:OmmkLam-rel}), 
altogether 
\beq
  \calO_{\pm k}= \Omega\Lambda^{\pm k} - \Lambda^{\mp k}\Omega 
  ~~ \text{for $k = 1,2,\ldots$}. 
  \label{3:OmpmLam-rel}
\eeq
The relation (\ref{2:Om-CK}) to the Cauchy kernel 
is also fundamental.  Actually, as hinted in the work 
of Fu and Nijhoff \cite{FN17,FN18,Fu18,FuThesis}, 
$\Omega$ is not uniquely determined 
by these conditions. For example, one can use 
the matrix 
\[
  \bar{\Omega} = - \sum_{i=0}^\infty E_{-i-1,i} 
\]
instead of $\Omega$.  In fact, $\bar{\Omega}$ 
satisfies the algebraic relations 
\beq
  \calO_{\pm k} 
  = \bar{\Omega}\Lambda^{\pm k} - \Lambda^{\mp k}\bar{\Omega} 
  ~~ \text{for $k = 1, 2,\ldots$}. 
  \label{3:OmbarpmLam-rel}
\eeq
Moreover, its two-variate generating function 
is the analytically continued Cauchy kernel:  
\beq
  \sum_{i,j\in\ZZ}\bar{\Omega}_{ij}w^iz^j 
  = - \sum_{j=0}^\infty\frac{z^j}{w^{j+1}} 
  = \frac{1}{z - w} 
  ~~ \text{for $|z| < |w|$}. 
  \label{3:Om-CK}
\eeq

Having the modified Omega matrix $\bar{\Omega}$,  
we can define the matrix $\bar{g}$ as 
\beq
  \bar{g} = E - U\bar{\Omega}. 
  \label{3:gbar-def}
\eeq
Just like the block structure (\ref{2:g-blocks}) 
of $g$, this matrix can be written as 
\beq
  \bar{g} 
  = \left(\begin{array}{c|c}
    \delta_{ij} & u_{i,-j-1} \\\hline
    0 & \delta_{ij} + u_{i,-j-1} 
    \end{array}\right). 
  \label{3:gbar-blocks}
\eeq
The central part looks as follows: 
\[
  \bar{g} = \begin{pmatrix}
      \ddots &\vdots &\vdots & \vdots &\vdots &  \\
      \cdots & 1 & 0 & u_{-2,-1} & u_{-2,-2} & \cdots\\
      \cdots & 0 & 1 & u_{-1,-1} & u_{-1,-2}& \cdots\\
      \cdots & 0 & 0 & 1+u_{0,-1} & u_{0,-2}& \cdots\\
      \cdots & 0 & 0 & u_{1,-1} & 1+u_{1,-2}& \cdots\\
         &\vdots&\vdots &\vdots &\vdots &\ddots
      \end{pmatrix}. 
\]
Note that the upper half of this matrix 
is nothing but the foregoing 
$\eta$-matrix $\bar{\eta}$. 

We can use $\bar{g}$ to capture 
the same Grassmannian structure as described 
by $\bar{\eta}$.  The reasoning is fully parallel 
to the description of the positive flows 
in terms of $\eta$ and $g$.  
Firstly, $\bar{g}$ satisfies 
the differential equations 
\beq
  \frac{\rd\bar{g}}{\rd t_{-k}} 
  = (\Lambda^{-k} - U\calO_{-k})\bar{g} - \bar{g}\Lambda^{-k}. 
  ~~ k =1,2,\ldots. 
  \label{3:gbar-tmk-eq}
\eeq
Secondly, the coefficient matrices 
$\Lambda^{-k} - U\calO_{-k}$ take 
the block-triangular form 
\beq
  \Lambda^{-k} - U\calO_{-k} 
  = \left(\begin{array}{c|c}
    \bar{\calC}_k & 0 \\\hline 
    * & * 
    \end{array}\right). 
  \label{3:Lmk-UOmk-block}
\eeq
Since the upper-left block of 
(\ref{3:Lmk-UOmk-block}) 
is the same matrix as the one showing up 
in (\ref{3:etabar-tmk-eq}) 
and the upper-half part of $\bar{g}$ 
is exactly $\bar{\eta}$, 
the equations (\ref{3:gbar-tmk-eq}) contain 
the equations (\ref{3:etabar-tmk-eq}) 
as a subsystem.  We have thus obtained 
substantially the same dynamical system 
on the Grassmannian $\overline{\Gr}$ 
in the different coset space realization 
\[
  \overline{\Gr} \simeq \overline{\calP}\backslash\GL(\ZZ). 
\]
The stabilizer $\overline{\calP} \subset \GL(\ZZ)$ 
consists of invertible matrices of the form 
\[
      p = \left(\begin{array}{c|c}
        * & 0 \\\hline
        * & * 
        \end{array}\right). 
\]

\subsection{Unification of positive and negative flows}

The foregoing description of the two Grassmannian 
structures can be unified in the extended system 
\beq
  \frac{\rd U}{\rd t_{\pm k}} 
  = \Lambda^{\pm k}U - U\Lambda^{\mp k} - U\calO_{\pm k}U, 
  ~~ k = 1,2,\ldots, 
  \label{3:U-tpmk-eq}
\eeq
of evolution equations for $U$. 
Note that the algebraic relations (\ref{3:OmpmLam-rel}) 
and (\ref{3:OmbarpmLam-rel}) among $\calO_k$'s, 
$\Omega$ and $\bar{\Omega}$ hold for both 
positive and negative values of $k$.  
They imply the compatibility 
of all flows defined by (\ref{3:U-tpmk-eq}). 

Moreover, the equations (\ref{2:g-tk-eq}) 
and (\ref{3:gbar-tmk-eq}) can be extended 
to both positive and negative flows, respectively. 
The derivation of these equations, too, 
is based on (\ref{3:OmpmLam-rel}) and 
(\ref{3:OmbarpmLam-rel}).  
The matrices $g,\bar{g}$ thus turn out 
to satisfy differential equations of the same form 
\beq
\begin{aligned}
  \frac{\rd g}{\rd t_{\pm k}} 
  &= (\Lambda^{\pm k} - U\calO_{\pm k})g 
     - g\Lambda^{\pm k},\\
  \frac{\rd\bar{g}}{\rd t_{\pm k}} 
  &= (\Lambda^{\pm k} - U\calO_{\pm k})\bar{g} 
     - \bar{g}\Lambda^{\pm k}. 
\end{aligned}
  \label{3:ggbar-tpmk-eq}
\eeq
Actually, as the following proposition implies, 
these equations for $g,\bar{g}$ are equivalent 
to the equations (\ref{3:U-tpmk-eq}) for $U$. 
Note that the mapping $U \mapsto (g,\bar{g})$ 
is bijective. 

\begin{prop}
The equations (\ref{3:ggbar-tpmk-eq}) for $g$ imply 
the equations (\ref{3:U-tpmk-eq}) for $U$.
\end{prop}

\proof 
Substituting $g = E - U\Omega$ in (\ref{3:ggbar-tpmk-eq}) 
and using (\ref{3:OmpmLam-rel}), we have 
the equation
\[
\begin{aligned}
  - \frac{\rd U}{\rd t_{\pm k}}\Omega 
  &= (\Lambda^{\pm k} - U\calO_{\pm k})(E - U\Omega) 
    - (E - U\Omega)\Lambda^{\pm k} \\
  &= - (\Lambda^{\pm k} - U\calO_{\pm k})U\Omega 
     - U(\calO_{\pm k} - \Omega\Lambda^{\pm k}) \\
  &= - (\Lambda^kU - U\calO_{\pm k}U - U\Lambda^{\mp k})\Omega. 
\end{aligned}
\]
This gives the $\ZZ\times\ZZ_{<0}$ part 
of (\ref{3:U-tpmk-eq}). In much the same way, 
we can derive the equation 
\[
  - \frac{\rd U}{\rd t_{\pm k}}\bar{\Omega} 
  = - (\Lambda^kU - U\calO_{\pm k}U 
       - U\Lambda^{\mp k})\bar{\Omega}  
\]
from the equations for $\bar{g}$.  This implies 
the the $\ZZ\times\ZZ_{\ge 0}$ part 
of (\ref{3:U-tpmk-eq}).
\qed

We can pack $g$ and $\bar{g}$ into a single matrix 
of the form 
\[
  \eta^{(2)} 
  = \left(\begin{array}{c|c}
    g & \bar{g}
    \end{array}\right)
\]
and rewrite the equations (\ref{3:ggbar-tpmk-eq}) as 
\beq
  \frac{\rd\eta^{(2)}}{\rd t_{\pm k}} 
  = (\Lambda^{\pm k} - U\calO_{\pm k})\eta^{(2)} 
    - \eta^{(2)}\diag(\Lambda^{\pm k},\Lambda^{\pm k}), 
  \label{3:eta2-tpmk-eq}
\eeq
where $\diag(\Lambda^{\pm k},\Lambda^{\pm k})$ denote 
the block-diagonal matrices 
\[
  \diag(\Lambda^{\pm k},\Lambda^{\pm k})
  = \left(\begin{array}{c|c}
    \Lambda^{\pm k} & 0 \\\hline
    0 & \Lambda^{\pm k}
    \end{array}\right).
\]
$\eta^{(2)}$ is a $\ZZ\times(\ZZ\sqcup\ZZ)$ 
matrix with the following block structure: 
\beq
  \eta^{(2)} 
  = \left(\begin{array}{c|c|c|c}
    \delta_{ij} - u_{i,-j-1} & 0 & \delta_{ij} & u_{i,-j-1}\\\hline
    - u_{i,-j-1} & \delta_{ij}& 0 & \delta_{ij} + u_{i,-j-1}
    \end{array}\right). 
  \label{3:eta2-blocks}
\eeq
The row index $i$ ranges over $i < 0$ 
in the upper blocks and $i \ge 0$ 
in the lower blocks. 
The column index $j$ ranges over 
$j < 0$, $j \ge 0$, $j < 0$ and $j \ge 0$, 
respectively, in the four blocks arranged 
from left to right.  Note that the previous 
two  matrices $\eta,\bar{\eta}$ are built 
into this larger matrix as 
\[
  \eta^{(2)} 
  = \left(\begin{array}{c|c}
    *    & \bar{\eta}\\\hline
    \eta & *
    \end{array}\right).
\]

The matrix $\eta^{(2)}$ may be thought of 
as representing a point of an open cell of 
yet another Sato Grassmannian 
\[
  \Gr^{(2)} 
  \simeq \GL(\ZZ)\backslash\Fr(\ZZ,\ZZ\sqcup\ZZ). 
\]
A more general representative of the same point 
of $\Gr^{(2)}$ satisfies the differential equations 
\beq
  \frac{\rd\eta^{(2)}}{\rd t_{\pm k}} 
  = \calC^{(2)}_k\eta^{(2)} 
    - \eta^{(2)}\diag(\Lambda^{\pm k},\Lambda^{\pm k}) 
  \label{3:eta2-tpmk-eq2}
\eeq
that allow gauge transformations 
\[
  \eta^{(2)} \to h\eta^{(2)},~~
  \calC^{(2)}_k \to h\calC^{(2)}_kh^{-1} 
    + \frac{\rd h}{\rd t_k}h^{-1}. 
\]

This kind of larger Grassmannians are used 
in the literature \cite{SS82,DJKM81,KvdL03} 
to describe the \textit{multi-component KP hierarchy}. 
In this sense, $\Gr^{(2)}$ is 
\textit{the two-component Sato Grassmannian}, 
and the equations (\ref{3:eta2-tpmk-eq}) 
describe the two-component KP hierarchy 
in geometric terms.  

Summarizing these considerations, 
we arrive at the following conclusion: 
\begin{theorem}
The equations (\ref{3:U-tpmk-eq}) for $U$ 
are equivalent to the two-component KP hierarchy 
realized as a dynamical system on the two-component 
Sato Grassmannian $\Gr^{(2)}$.  All matrix elements 
$u_{ij}$ of $U$ can be identified with 
affine coordinates of an open cell of $\Gr^{(2)}$.
\end{theorem}

\begin{remark}
A matrix of almost the same form as $\eta^{(2)}$ 
shows up in \textit{the 2D Toda hierarchy}
\cite{Takasaki90LMP}
\footnote{To be more precise, the 2D Toda hierarchy 
considered therein is a supersymmetric analogue.}
(see Appendix B). 
It is well known that the 2D Toda hierarchy 
is closely related to the two-component 
KP hierarchy \cite{UT84,Takasaki84}. 
The two sets of time variables $\bst = (t_1,t_2,\ldots)$ 
and $\bar{\bst} = (\bar{t}_1,\bar{t}_2,\ldots)$ 
of the 2D Toda hierarchy correspond 
to two sets of time variables 
of the two-component KP hierarchy.  
A difference with the 2D Toda hierarchy 
is that the $\eta$-matrix of the 2D Toda hierarchy 
depends on the lattice coordinate $s$ as well.  
The matrix $\eta^{(2)}$ amounts to 
the $\eta$-matrix $\eta(s)^{(2)}$ 
of the 2D Toda hierarchy at $s = 0$
\footnote{The multi-component KP hierarchy itself 
can be formulated to accommodate several 
discrete variables \cite{DJKM81,KvdL03}. 
These discrete variables are related to charges 
of fermions in the fermionic description 
of the multi-component KP hierarchy.}. 
\end{remark}

\begin{remark}
The evolution equation (\ref{3:eta2-tpmk-eq2}) 
can be modified as 
\[
  \frac{\rd\eta^{(2)\prime}}{\rd t_k} 
  = \calC_k^{(2)}\eta^{(2)\prime} 
    - \eta^{(2)\prime}\diag(\Lambda^k,0)
\]
for the positive flows and
\[
  \frac{\rd\eta^{(2)\prime}}{\rd t_{-k}} 
  = \calC_{-k}^{(2)}\eta^{(2)\prime} 
    - \eta^{(2)\prime}\diag(0,\Lambda^{-k})
\]
for the negative flows 
by the simple transformation 
\[
  \eta^{(2)\prime} 
  = \eta^{(2)}\exp\diag
    \left(-\sum_{k=1}^\infty t_{-k}\Lambda^{-k}, 
     -\sum_{k=1}^\infty t_k\Lambda^k\right). 
\]
It is this $\eta$-matrix $\eta^{(2)\prime}$ 
rather than $\eta^{(2)}$ that shows up in the usual 
formulation of the two-comonent KP hierarchy. 
\end{remark}

\subsection{Perspective from wave functions}

Let us consider the relationship with 
the two-component KP hierarchy 
in the language of wave functions. 
To this end, we modify the previous definition 
(\ref{2:u(z)-def}) of $\bsu(z)$ and newly 
define its partner $\bar{\bsu}(z)$ as 
\[
  \bsu(z) = g\bsc(z)\rho(z), ~~
  \bar{\bsu}(z) = \bar{g}\bsc(z)\rho(z), 
\]
where 
\[
  \rho(z) = \exp\left(\sum_{k=1}^\infty 
            (t_kz^k + t_{-k}z^{-k})\right). 
\]
The plane wave factor $\rho(z)$ is modified 
to accommodate both positive and negative flows.  
As a consequence of the differential equations 
(\ref{3:ggbar-tpmk-eq}) satisfied by $g$ and $\bar{g}$, 
we have the linear differential equations 
\beq
\begin{aligned}
  \frac{\rd\bsu(z)}{\rd t_{\pm k}} 
  &= (\Lambda^{\pm k} - U\calO_{\pm k})\bsu(z),\\
  \frac{\rd\bar{\bsu}(z)}{\rd t_{\pm k}} 
  &= (\Lambda^{\pm k} - U\calO_{\pm k})\bar{\bsu}(z)  
\end{aligned}
  \label{3:uubar(z)-tpmk-eq}
\eeq
for $k = 1,2,\ldots$. 

The components $u_i(z),\bar{u}_i(z)$ of 
$\bsu(z),\bar{\bsu}(z)$ can be expressed as 
\[
\begin{aligned}
  u_i(z) &= \left(z^i - \sum_{j=0}^\infty u_{ij}z^{-j-1}
           \right)\rho(z),\\
  \bar{u}_i(z) &= \left(z^i + \sum_{j=0}^\infty u_{i,-j-1}z^j
           \right)\rho(z). 
\end{aligned}
\]
As indicated by the block structure (\ref{3:eta2-blocks}) 
of $\eta^{(2)}$, it will be natural to divide 
these wave functions into the four groups 
$\{u_i(z)\}_{i=0}^\infty$, $\{u_{-i-1}(z)\}_{i=0}^\infty$, 
$\{\bar{u}_i(z)\}_{i=0}^\infty$ and 
$\{\bar{u}_{-i-1}(z)\}_{i=0}^\infty$ 
with the following leading members: 
\[
\begin{aligned}
  u_0(z) &= \left(1 
     - \sum_{j=0}^\infty u_{0,j}z^{-j-1}\right)\rho(z),\\
  \bar{u}_0(z) &= \left((1 + u_{0,-1}) 
     + \sum_{j=1}^\infty u_{0,-j-1}z^j\right)\rho(z),\\
  u_{-1}(z) &= \left((1 - u_{-1,0})z^{-1} 
     - \sum_{j=1}^\infty u_{-1,j}z^{-j-1}\right)\rho(z),\\
  \bar{u}_{-1}(z) &= \left(z^{-1} 
     + \sum_{j=0}^\infty u_{-1,-j-1}z^j\right)\rho(z).       
\end{aligned}
\]

As we show below, the $2 \times 2$ matrix 
\beq
  \Psi_{2\times 2}(z) = \begin{pmatrix}
            u_0(z) & \bar{u}_0(z) \\
            u_{-1}(z) & \bar{u}_{-1}(z) 
            \end{pmatrix}
  \label{3:Psi22-def}
\eeq
play the role of a matrix-valued wave function 
for the two-component KP hierarchy \cite{DJKM81,KvdL03}.
This is another way to verify that the extended $U$-system 
is related to the two-component KP hierarchy.  
A clue towards this interpretation is the following 
generalization of the relation 
(\ref{2:ui(z)u0(z)-rel}) between $u_i(z)$ and $u_0(z)$. 
Let $\rd_x$ denote the sum of $\rd/\rd t_1$ 
and $\rd/\rd t_{-1}$:
\footnote{Equivalently, we can shift $t_1$ and 
$t_{-1}$ by $x$ and consider the derivative 
operator in $x$.}
\[
  \rd_x = \rd/\rd t_1 + \rd/\rd t_{-1}.
\]

\begin{lemma}
The wave functions $u_i(z)$, $\bar{u}_i(z)$, 
$u_{-i-1}(z)$, $\bar{u}_{-i-1}(z)$ for $i \ge 1$ 
can be expressed by $u_0(z)$, $\bar{u}_0(z)$, 
$u_{-1}(z)$, $\bar{u}_{-1}(z)$ as 
\beq
\begin{aligned}
  u_i(z) &= P_iu_0(z) + Q_iu_{-1}(z),\\
  \bar{u}_i(z) &= P_i\bar{u}_0(z) + Q_i\bar{u}_{-1}(z),\\
  u_{-i-1}(z) &= R_iu_0(z) + S_iu_{-1}(z),\\
  \bar{u}_{-i-1}(z) &= R_i\bar{u}_0(z) + S_i\bar{u}_{-1}(z), 
\end{aligned}
  \label{3:higher-uubar(z)}
\eeq
where $P_i,Q_i,R_i$, $S_i$ are 
differential operators of the form 
\[
\begin{aligned}
  P_i &= \rd_x^i + O(\rd_x^{i-1}),\\
  Q_i &= - (1 + u_{0,-1})\rd_x^{i-1} + O(\rd_x^{i-2}),\\
  R_i &= - (1- u_{-1,0})\rd_x^{i-1} + O(\rd_x^{i-2}),\\
  S_i &= \rd_x^i + O(\rd_x^{i-1}), 
\end{aligned}
\]
where $O(\rd_x^n)$ stands for a differential 
operator of order less than or equal to $n$. 
\end{lemma}

\proof
In terms of the components, the equations 
(\ref{3:uubar(z)-tpmk-eq}) for $k = \pm 1$  
can be written as 
\[
\begin{aligned}
  \frac{\rd u_i(z)}{\rd t_1} 
  &= u_{i+1}(z) - u_{i,0}u_0(z),\\
  \frac{\rd\bar{u}_i(z)}{\rd t_1} 
  &= \bar{u}_{i+1}(z) - u_{i,0}\bar{u}_0(z),\\
  \frac{\rd u_i(z)}{\rd t_{-1}} 
  &= u_{i-1}(z) + u_{i,-1}\bar{u}_{-1}(z),\\
  \frac{\rd\bar{u}_i(z)}{\rd t_{-1}} 
  &= \bar{u}_{i-1}(z) + u_{i,-1}\bar{u}_{-1}(z). 
\end{aligned}
\]
Adding the $t_1$-derivative and 
the $t_{-1}$-derivative, we obtain the equations 
\beq
 \begin{aligned}
  \rd_xu_i(z) &= u_{i+1}(z) + u_{i-1}(z) 
    - u_{i,0}u_0(z) + u_{i,-1}u_{-1}(z), \\
  \rd_x\bar{u}_i(z) &= \bar{u}_{i+1}(z) + \bar{u}_{i-1}(z) 
    - u_{i,0}\bar{u}_0(z) + u_{i,-1}\bar{u}_{-1}(z). 
\end{aligned}
  \label{3:uubar-x-eq}
\eeq
Letting $i = 0$ and $i = -1$ and moving 
some terms, we have 
\[
\begin{aligned}
  u_1(z) &= (\rd_x + u_{0,0})u_0(z) 
    - (1 + u_{0,-1})u_{-1}(z) ,\\
  \bar{u}_1(z) &= (\rd_x + u_{0,0})\bar{u}_0(z)
    - (1 + u_{0,-1})\bar{u}_{-1}(z) ,\\
  u_{-2}(z) &= - (1 - u_{-1,0})u_0(z)
    + (\rd_x - u_{-1,-1})u_{-1}(z) ,\\
  \bar{u}_{-2}(z) &= - (1 - u_{-1,0})\bar{u}_0(z)
    + (\rd_x - u_{-1,-1})\bar{u}_{-1}(z) . 
\end{aligned}
\]
This gives (\ref{3:higher-uubar(z)}) 
for $i = 1$.  Now letting $i = 1$ and $i = -2$ 
in (\ref{3:uubar-x-eq}) and 
using the foregoing expression of 
$u_1(z)$, $\bar{u}_1(z)$, $u_{-2}(z)$, $\bar{u}_{-2}(z)$, 
we can similarly derive (\ref{3:higher-uubar(z)}) 
for $i = 2$.  Repeating these manipulations, 
one can derive (\ref{3:higher-uubar(z)}) 
for all $i \ge 1$. 
\qed

\begin{theorem}
The matrix wave function (\ref{3:Psi22-def}) 
satisfies the linear differential equations 
\beq
  \frac{\rd\Psi_{2\times 2}(z)}{\rd t_{\pm k}} 
  = \begin{pmatrix}
    A_{\pm k} & B_{\pm k} \\
    C_{\pm k} & D_{\pm k} 
    \end{pmatrix}
    \Psi_{2\times 2}(z)  
  \label{3:Psi-tpmk-eq}
\eeq
for the positive and negative flows. 
$A_{\pm k},B_{\pm k},C_{\pm k},D_{\pm k}$ 
are differential operators of the form 
\[
\begin{aligned}
  A_k &= \rd_x^k + O(\rd_x^{k-1}),\\
  B_k &= -(1 + u_{0,-1})\rd_x^{k-1} + O(\rd_x^{k-2}),\\
  C_k &= (1 - u_{-1,0})\rd_x^{k-1} + O(\rd_x^{k-2}),\\
  D_k &= O(\rd_x^{k-2})
\end{aligned}
\]
and 
\[
\begin{aligned}
  A_{-k} &= O(\rd_x^{k-2}),\\
  B_{-k} &= (1 + u_{0,-1})\rd_x^{k-1} + O(\rd_x^{k-2}),\\
  C_{-k} &= - (1 - u_{-1,0})\rd_x^{k-1} + O(\rd_x^{k-2}),\\
  D_{-k} &= \rd_x^k + O(\rd_x^{k-1}). 
\end{aligned}
\]
\end{theorem}

\proof 
The $0$-th components of (\ref{3:uubar(z)-tpmk-eq}) 
read 
\[
\begin{aligned}
  \frac{\rd u_0(z)}{\rd t_k}
  &= u_k(z) - \sum_{l=0}^{k-1}u_{0,l}u_{k-l-1}(z),\\
  \frac{\rd u_0(z)}{\rd t_{-k}}
  &= u_{-k}(z) + \sum_{l=1}^ku_{0,-l}u_{l-k-1}(z). 
\end{aligned}
\]
Using (\ref{3:higher-uubar(z)}), we can rewrite 
each term of the right hand side into an expression 
consisting solely of $u_0(z)$ and $u_{-1}(z)$.  
This leads to the equation 
\[
  \frac{\rd u_0(z)}{\rd t_k} = A_iu_0(z) + B_iu_{-1}(z), 
\]
where $A_i$ and $B_i$ are differential operators 
of the form 
\[
  A_i = \rd_x^k + O(\rd_x^{k-1}),~~
  B_i = - (1 + u_{0,-1})\rd_x^{k-1} + O(\rd_x^{k-2}). 
\]
This gives the upper-left component of 
the matrix system (\ref{3:Psi-tpmk-eq}).  
The other components can be verified in the same way. 
\qed 

The equations (\ref{3:Psi-tpmk-eq}) may be thought of 
as auxiliary linear equations of the two-component 
KP hierarchy. In order to make a comparison with 
the usual formulation \cite{SS82,DJKM81,KvdL03} 
more precise, we can modify the $2 \times 2$ matrix-valued 
wave function as 
\[
  \Psi_{2\times 2}(z) 
  = \begin{pmatrix}
    u_0(z) & \bar{u}_0(z^{-1})z^{-1}\\
    u_{-1}(z) & \bar{u}_{-1}(z^{-1})z^{-1} 
    \end{pmatrix}. 
\]
This matrix-valued wave function 
satisfies the same linear differential 
equations as (\ref{3:Psi-tpmk-eq}).

\subsection{Special solutions and wave functions}

The special solution (\ref{2:CM-uij}) 
of (\ref{2:U-tk-eq}) can be extended 
to a solution of (\ref{3:U-tpmk-eq}) 
by modifying $\rho(z)$ and $\sigma(z)$ 
to depend on both $t_k$'s and $t_{-k}$'s as 
\[
\begin{aligned}
  \rho(z) &= \exp\left(\sum_{k=1}^\infty 
            (t_kz^k + t_{-k}z^{-k})\right),\\
  \sigma(z) &= \exp\left(- \sum_{k=1}^\infty 
            (t_kz^k + t_{-k}z^{-k})\right). 
\end{aligned}
\]
The Cauchy matrix $\calM$ and the vectors $\bsr,\bss$ 
are thereby defined as 
\[
 \calM = \left(\frac{\sigma(\lambda_i)\rho(\kappa_j)}
   {\kappa_j - \lambda_i}\right)_{1\le i\le M,1\le j\le N}
\]
and 
\[
  \bsr = \left(\rho(\kappa_i)\right)_{i=1}^N,~~
  \bss = \left(\sigma(\lambda_i)\right)_{i=1}^M. 
\]

One can verify by direct calculations   
that the modified $u_{ij}$'s satisfy 
the extended $U$-system (\ref{3:U-tpmk-eq}). 
The following lemma plays a key role therein. 

\begin{lemma}
$\calM$, $\bsr$ and $\bss$ satisfy 
the differential equations 
\beq
\begin{aligned}
  \frac{\rd\calM}{\rd t_k} 
  &= \sum_{l=0}^{k-1}L^l\bss\tp{\bsr}K^{k-l-1},\\
  \frac{\rd\bsr}{\rd t_k} &= K^k\bsr,\\
  \frac{\rd\bss}{\rd t_k} &= L^k\bss
  \label{3:CM-tk-eq}
\end{aligned}
\eeq
for the positive flows and 
\beq
\begin{aligned}
  \frac{\rd\calM}{\rd t_{-k}} 
  &= - \sum_{l=1}^kL^{-l}\bss\tp{\bsr}K^{l-k-1},\\
  \frac{\rd\bsr}{\rd t_{-k}} &= K^{-k}\bsr,\\
  \frac{\rd\bss}{\rd t_{-k}} &= L^{-k}\bss
\end{aligned}
  \label{3:CM-tmk-eq}
\eeq
for the negative flows. 
\end{lemma}

\proof
The derivatives of $\calM$ can be written as 
\[
  \frac{\rd\calM}{\rd t_{\pm k}} 
  = \left(
    \frac{\kappa_j^{\pm k} - \lambda_i^{\pm k}}{\kappa_j - \lambda_i}
    \sigma(\lambda_i)\rho(\kappa_j)\right)_{1\le i\le M,1\le j\le N}. 
\]
Applying the identities 
\[
  \frac{\kappa_j^k - \lambda_i^k}{\kappa_j - \lambda_i}
  = \sum_{l=0}^{k-1}\kappa_j^l\lambda_i^{k-l-1},~~
  \frac{\kappa_j^{\pm k} - \lambda_i^{\pm k}}{\kappa_j - \lambda_i}
  = - \sum_{l=1}^k\kappa_j^{-l}\lambda_i^{l-k-1} 
\]
to the right hand side, we obtain 
the differential equations of $\calM$. 
The differential equations of $\bsr$ and $\bss$ 
are an immediate consequence of the definition.  
\qed

\begin{theorem}
The matrix $U = (u_{ij})_{i,j\in\ZZ}$ with 
the matrix elements 
\[
  u_{ij} 
  = \tp{\bsr}K^i(E_N + A\calM)^{-1}AL^j\bss 
  = \tp{\bsr}K^iA(E_M + \calM A)^{-1}L^j\bss
\]
satisfies the extended $U$-system (\ref{3:U-tpmk-eq}). 
\end{theorem}

\proof 
Using (\ref{3:CM-tk-eq}), we can calculate 
the $t_k$-derivative of $u_{ij}$ as follows: 
\[
\begin{aligned}
  \frac{\rd u_{ij}}{\rd t_k} 
  &= \frac{\rd\tp{\bsr}}{\rd t_k}K^i(E_N + A\calM)^{-1}AL^j\bss
   + \tp{\bsr}K^i(E_N + A\calM)^{-1}AL^j\frac{\rd\bss}{\rd t_k}\\
  &~~~~\mbox{} 
   - \tp{\bsr}K^iA(E_M + \calM A)^{-1}
     \frac{\rd\calM}{\rd t_k}
     (E_M + \calM A)^{-1}AL^j\bss \\
  &= \tp{\bsr}K^{i+k}(E_N + A\calM)^{-1}AL^j\bss 
     - \tp{\bsr}K^i(E_N + A\calM)^{-1}AL^{j+k}\bss \\
  &~~~~\mbox{} 
     - \sum_{l=0}^{k-1}\tp{\bsr}K^iA(E_M + \calM A)^{-1}
       L^l\bss\tp{\bsr}K^{k-l-1}
       (E_M + \calM A)^{-1}AL^j\bss \\
  &= u_{i+k,j} - u_{i,j+k} - \sum_{l=0}^{k-1}u_{il}u_{k-l-1,j}. 
\end{aligned}
\]
This is exactly the equations (\ref{3:U-tpmk-eq}) 
for the positive flows.  In much the same way 
using (\ref{3:CM-tmk-eq}), we can derive 
the equations (\ref{3:U-tpmk-eq}) 
for the negative flows. 
\qed

\section{Constrained multi-component $U$-matrix}

\subsection{Constraints for multi-component $U$-matrix}

Direct linearization of the AKNS and ASDYM 
hierarchies employs the \textit{constrained} 
and \textit{multi-component} $U$-matrix 
\[
  U = (u_{ij})_{i,j\in\ZZ}
\]
that consists of $r \times r$ blocks 
($r = 2$ for the AKNS hierarchy)
\[
  u_{ij} = (u_{ij,\alpha\beta})_{\alpha,\beta=1}^r, 
\]
and satisfies the constraint 
\beq
  (E_r\otimes\Lambda)U - U(E_r\otimes\Lambda^{-1}) 
  - U(E_r\otimes\calO_1)U = 0. 
  \label{4:U-constraint}
\eeq
For an $r \times r$ matrix $A$ and a $\ZZ\times\ZZ$ 
matrix $B = (b_{ij})_{i,j\in\ZZ}$, multiplication by 
$A\otimes B$ means that 
\[
  (A\otimes B)U = \left(\sum_{k\in\ZZ}b_{ik}Au_{kj}\right), ~~
  U(A\otimes B) = \left(\sum_{k\in\ZZ}u_{ik}Ab_{kj}\right). 
\]
The constraint (\ref{4:U-constraint}) thus becomes 
the quadratic relations 
\beq
  u_{i+1,j} - u_{i,j+1} - u_{i,0}u_{0,j} = 0  
  ~~ \text{for $i,j \in \ZZ$}. 
  \label{4:uij-constraint}
\eeq
This constraint and its higher rank analogues 
are known to show up in the Cauchy matrix approach 
to special solutions as well 
\cite{Zhao18,LQYZ22,LQZ23}. 

The quarter $u_{ij}$, $i,j\ge 0$, of $u_{ij}$'s 
can be identified with $r \times r$ matrix-valued 
affine coordinates $w_{ij}$ of the top cell 
of the $r$-component Sato Grassmannian $\Gr^{(r)}$ 
\cite{SS82,DJKM81,KvdL03}. 
Its general point is represented by the matrices 
\[
  \xi = \left(\begin{array}{c}
         \delta_{ij}E_r\\\hline 
         w_{i,-j-1}
         \end{array}\right)_{i\in\ZZ,j\in\ZZ_{<0}},~~
  \eta = \left(\begin{array}{c|c}
         - w_{i,-j-1} & \delta_{ij}E_r 
         \end{array}\right)_{i\in\ZZ_{\ge 0},j\in\ZZ}
\]
consisting of $r \times r$ blocks. 
The blocks $w_{ij}$ are required 
to satisfy the constraint 
\beq
  w_{i+1,j} - w_{i,j+1} - w_{i,0}w_{0,j} = 0 
  ~~\text{for $i,j \ge 0$}, 
  \label{4:wij-constraint}
\eeq
which can be rewritten as 
\beq
  (E_r\otimes\Lambda)\xi = \xi\calB,~~
  \eta(E_r\otimes\Lambda) = \calC\eta, 
  \label{4:xieta-constraint}
\eeq
where $\calB$ and $\calC$ take 
substantially the same form as 
$\calB_1$ and $\calC_1$ in (\ref{2:xi-tk-eq}) 
and (\ref{2:eta-tk-eq}) except that 
the `matrix elements' become $r \times r$ blocks.  

\begin{remark}
The construction of the two-component Grassmannian 
$\Gr^{(2)}$ looks apparently different from 
that of Section 3. Actually, they are equivalent, 
the index set of rows and columns of the matrices 
being reorganized from $\ZZ\sqcup\ZZ$ 
to $\{1,2\}\times\ZZ$. 
\end{remark}

The affine coordinates $w_{ij}$ are not independent 
under the constraint (\ref{4:wij-constraint}).  
One can choose $w_{i0}$'s or $w_{0j}$'s as 
free parameters; the other $w_{ij}$'s 
are uniquely determined by these parameters.  
This fact can be stated in a more explicit form 
as follows.  Let $G(z,y)$ denote the two-variate 
generating function 
\beq
  G(z,y) = \sum_{i,j\ge 0}w_{ij}y^{-i-1}z^{-j-1}
  \label{4:G-wij}
\eeq
of $w_{ij}$'s and $W(z)$ the one-variate 
generating function 
\beq
  W(z) = E_r - \sum_{j=0}^\infty w_{0j}z^{-j-1}
  \label{4:W(z)-w0j}
\eeq
of $w_{0j}$'s. 

\begin{prop}[\cite{Takasaki89RMP,Takasaki84CMP}]
Under the constraint (\ref{4:wij-constraint}), 
$G(z,y)$ can be expressed by $W(z)$ and $W(y)^{-1}$ as 
\[
  G(z,y) = \frac{W(y)^{-1}W(z) - E_r}{z - y}. 
\]
Moreover, $W(z)^{-1}$ can be written as 
\[
  W(z)^{-1} = E_r + \sum_{i=0}^\infty w_{i0}z^{-i-1}. 
\]
Conversely, given an arbitrary power series $W(z)$ 
of $z^{-1}$ of the foregoing form, the coefficients 
$w_{ij}$ of the generating function (\ref{4:G-wij}) 
satisfy the constraint (\ref{4:wij-constraint}). 
\end{prop}

Another approach to $\Gr^{(r)}$ can be achieved 
by the following analogue of (\ref{2:g-def}): 
\[
  g = E - U(E_r\otimes\Omega) 
    = \left(\begin{array}{c|c}
    \delta_{ij}E_r - u_{i,-j-1} & 0 \\\hline
    - u_{i,-j-1} & \delta_{ij}E_r 
    \end{array}\right). 
\]
This matrix represents a point of $\Gr^{(r)}$ 
in the coset space description 
\[
  \Gr^{(r)} 
  \simeq \calP^{(r)}\backslash\GL(\ZZ^{\sqcup r}), ~~
  \ZZ^{\sqcup r} = \ZZ\sqcup\cdots\sqcup\ZZ ~ 
  (\text{$r$-fold disjoint union}),  
\]
with a maximal parabolic subgroup $\calP^{(r)}$ 
of $\GL(\ZZ^{\sqcup r})$. 
The foregoing matrix $\eta$ is embedded 
in the lower half of this matrix.  
In much the same way as the proof of (\ref{2:g-tk-eq}), 
one can confirm that the constraint (\ref{4:U-constraint}) 
implies the constraint 
\beq
  (E_r\otimes\Lambda - U(E_r\otimes\calO_1))g 
  = g(E_r\otimes\Lambda) 
  \label{4:g-constraint}
\eeq
for $g$.  The multiplier on the left side 
has the block structure 
\[
  E_r\otimes\Lambda - U(E_r\otimes\calO_1)
  = \left(\begin{array}{c|c}
    * & * \\\hline
    0 & \calC  
    \end{array}\right) 
\]
and belongs to the Lie algebra of $\calP^{(r)}$.

\subsection{Implications of constraint for wave functions}

The foregoing matrix $g$ generates 
the `vector-valued' wave function 
\[
  \bsu(z) = (u_i(z))_{i\in\ZZ} = g\bsc(z)\rho(z), 
\]
where $\bsc(z)$ is redefined as 
\[
  \bsc(z) = (E_rz^i)_{i\in\ZZ}, 
\]
and $\rho(z)$ is a system-dependent $r \times r$ 
matrix to be specified later.  
$\bsu(z)$ is an $\infty\times r$ matrix 
with $r \times r$ blocks $u_i(z)$, $i \in \ZZ$.  
More explicitly, 
\[
  u_i(z) 
  = \left(E_rz^i - \sum_{j=0}^\infty u_{ij}z^{-j-1}\right)\rho(z). 
\]
Although written in an apparently uniform form, 
$u_i(z)$'s have different features 
for the cases of $i \ge 0$ and $i < 0$. 
Representatives of these two cases 
are $u_0(z)$ and $u_{-1}(z)$: 
\[
\begin{aligned}
  u_0(z) 
  &= \left(E_r - \sum_{j=0}^\infty u_{0j}z^{-j-1}
     \right)\rho(z),\\
  u_{-1}(z) 
  &= \left((E_r - u_{-1,0})z^{-1} 
     - \sum_{j=1}^\infty u_{-1,j}z^{-j-1}\right)\rho(z). 
\end{aligned}
\]
The foregoing generating function (\ref{4:W(z)-w0j}) 
of $w_{0j}$'s amounts to the amplitude part 
of $u_0(z)$: 
\beq
  u_0(z) = W(z)\rho(z). 
  \label{4:u0(z)-W(z)}
\eeq

The constraint (\ref{4:uij-constraint}) 
for $u_{ij}$'s implies the linear relations 
\beq
  u_{i+1}(z) - zu_i(z) - u_{i0}u_0(z) = 0 
  \label{4:ui(z)-constraint}
\eeq
among $u_i(z)$'s.  We can thereby express 
$u_i(z)$ and $u_{-i}(z)$ for $i \ge 1$ 
in terms of $u_0(z)$ as follows. 

\begin{prop}
$u_i(z)$ and $u_{-i}(z)$ for $i \ge 1$ 
can be expressed as 
\beq
  u_i(z) = \left(E_rz^i 
    + \sum_{l=0}^{i-1}u_{l,0}z^{i-l-1}\right)u_0(z)
  \label{4:ui(z)-u0(z)}
\eeq
and
\beq
  u_{-i}(z) = \left((E_r - u_{-1,0})z^{-i} 
     - \sum_{l=2}^i u_{-l,0}z^{l-i-1}\right)u_0(z)
  \label{4:u(-i)(z)-u0(z)}. 
\eeq
\end{prop}

\proof
We can calculate $u_i(z)$'s recursively  
by using (\ref{4:ui(z)-constraint}) as 
\[
\begin{aligned}
  u_1(z) 
  &= zu_0(z) + u_{00}u_0(z)\\
  &= (E_rz + u_{00})u_0(z),\\
  u_2(z) 
  &= zu_1(z) + u_{10}u_0(z) \\
  &= z(E_rz + u_{00}u_0(z)) + u_{10}u_0(z) \\
  &= \left(E_rz^2 + u_{00}z + u_{10}\right)u_0(z),
    ~~ \text{etc}. 
\end{aligned}
\]
This leads to (\ref{4:ui(z)-u0(z)}). 
In much the same way, $u_{-i}(z)$'s 
can be obtained by rewriting 
(\ref{4:ui(z)-constraint}) into the form 
\[
  u_{i-1}(z) = z^{-1}u_i(z) - z^{-1}u_{i-1,0}u_0(z) 
\]
as 
\[
\begin{aligned}
  u_{-1}(z) 
  &= z^{-1}u_0(z) - z^{-1}u_{-1,0}u_0(z)\\
  &= (E_r - u_{-1,0})z^{-1}u_0(z),\\
  u_{-2}(z) 
  &= z^{-1}u_{-1}(z) - z^{-1}u_{-2,0}u_0(z) \\
  &= (E_r - u_{-1,0})z^{-2}u_0(z) - z^{-1}u_{-2,0}u_0(z) \\
  &= \left((E_r - u_{-1,0})z^{-2} - u_{-2,0}z^{-1}\right)u_0(z), 
     ~~ \text{etc}. 
\end{aligned}
\]
This yields (\ref{4:u(-i)(z)-u0(z)}). 
\qed

We can also introduce the analogue 
\[
  \bar{g} 
  = E - U(E_r\otimes\bar{\Omega}) 
  = \left(\begin{array}{c|c}
    \delta_{ij}E_r & u_{i,-j-1} \\\hline
    0 & \delta_{ij}E_r + u_{i,-j-1} 
    \end{array}\right). 
\]
of (\ref{3:gbar-def}).  This matrix $\bar{g}$ 
generates the second set $\bar{u}_i(z)$, $i \in \ZZ$, 
of wave functions as 
\[
  \bar{\bsu}(z) = (\bar{u}_i(z))_{i\in\ZZ} 
  = \bar{g}\bsc(z)\rho(z). 
\]
In particular, we have 
\[
\begin{aligned}
  \bar{u}_0(z) &= \left((E_r + u_{0,-1}) 
    + \sum_{j=1}^\infty u_{0,-j-1}z^j\right)\rho(z),\\
  \bar{u}_{-1}(z) &= \left(E_rz^{-1}
    + \sum_{j=0}^\infty\bar{u}_{-1,-j-1}z^j\right)\rho(z).
\end{aligned}
\]
The constraint (\ref{4:uij-constraint}) 
for $u_{ij}$'s implies the following relations 
among $\bar{u}_i(z)$'s:

\begin{prop}
$\bar{u}_i(z)$ and $\bar{u}_{-i}(z)$, $i \ge 1$, 
can be expressed as 
\beq
  \bar{u}_i(z) = \left(E_rz^i 
    + \sum_{l=0}^{i-1}u_{l0}z^{i-l-1}\right)\bar{u}_0(z)
  \label{4:ubari(z)-ubar0(z)}
\eeq
and 
\beq
  \bar{u}_{-i}(z) = \left((E_r - u_{-1,0})z^{-i} 
    - \sum_{l=2}^i u_{-l,0}z^{l-i-1}\right)\bar{u}_0(z). 
  \label{4:ubar(-i)(z)-ubar0(z)}
\eeq
\end{prop}

\section{Direct linearization of AKNS and ASDYM hierarchies}

\subsection{$U$-system for AKNS hierarchy}

As a prototype for considering the ASDYM hierarchy, 
let us present the $U$-system 
for the AKNS hierarchy and its implications 
omitting details.  

In this case, $r$ is specialized to $r = 2$. 
Let $\bst = (t_1,t_2,\ldots)$ be the set 
of time variables and $a$ a $2 \times 2$ 
semi-simple non-zero constant matrix, 
e.g., 
\[
  a = \sigma_3 = \diag(1,-1). 
\]
Direct linearization of the AKNS hierarchy 
is based on the evolution equations 
\beq
  \frac{\rd U}{\rd t_k} 
  = (a\otimes\Lambda^k)U - U(a\otimes\Lambda^{-k}) 
    - U(a\otimes\calO_k)U, ~~
  k = 1,2,\ldots
  \label{5:U-tk-eq}
\eeq
for the matrix $U$ under the constraint 
(\ref{4:U-constraint}).  

As in the case of the one-component $U$-system, 
these equations can be transformed 
to the linear equations 
\beq
  \frac{\rd C}{\rd t_k} 
  = (a\otimes\Lambda^k)C - C(a\otimes\Lambda^{-k})
  \label{5:C-tk-eq}
\eeq
for the matrix $C = (c_{ij})_{i,j\in\ZZ}$ 
of $2 \times 2$ blocks $c_{ij}$ by the relation 
\[
  U = (E - U(E_2\otimes\Omega))C
\]
between $U$ and $C$.  The constraint 
(\ref{4:U-constraint}) for $U$ 
turns into the constraint 
\beq
  (E_2\otimes\Lambda)C - C(E_2\otimes\Lambda^{-1}) = 0 
  \label{5:C-constraint}
\eeq
for $C$. 

The equations (\ref{5:U-tk-eq}) for $U$ 
become the evolution equations 
\beq
  \frac{\rd u_{ij}}{\rd t_k} 
  = au_{i+k,j} - u_{i,j+k}a - \sum_{l=0}^{k-1}u_{il}au_{k-l-1,j},  
  ~~i,j\in\ZZ 
  \label{5:uij-tk-eq}
\eeq
for $u_{ij}$'s. Among these equations, 
those for $i,j \ge 0$ form a closed system 
for $\{u_{ij}\}_{i,j=0}^\infty$. 
Identifying these $u_{ij}$'s with the $2 \times 2$ 
matrix-valued affine coordinates $\{w_{ij}\}_{i,j=1}^\infty$ 
of $\Gr^{(2)}$, we obtain the evolution equations 
\beq
  \frac{\rd w_{ij}}{\rd t_k} 
  = aw_{i+k,j} - w_{i,j+k}a - \sum_{l=1}^{k-1}w_{il}aw_{k-l-1,j}
  \label{5:wij-tk-eq}
\eeq
for $w_{ij}$'s.  These equations can be cast 
into the equations 
\beq
  \frac{\rd\xi}{\rd t_k} 
  = (a\otimes\Lambda^k)\xi - \xi\calB_k,~~
  \frac{\rd\eta}{\rd t_k} 
  = \calC_k\eta - \eta(a\otimes\Lambda^k). 
  \label{5:xieta-tk-eq}
\eeq
for the matrices $\xi,\eta$ 
introduced in the last section. 
$\calB_k$ and $\calC_k$ are natural generalizations 
\[
\begin{aligned}
  \calB_k &= \left(\begin{array}{cl}
         \delta_{i+k,j}E_2 & (i<-k)\\\hline
         w_{i+k,-j-1} & (-k\le i<0)
         \end{array}\right)_{i,j\in\ZZ_{<0}},\\
  \calC_k &= \left(\begin{array}{c|c}
         - w_{i,k-j-1}~~(0\le j<k) & \delta_{i,j-k}E_2~~(j\ge k)
         \end{array}\right)_{i,j\in\ZZ_{\ge 0}} 
\end{aligned}
\]
of those in (\ref{2:xi-tk-eq}) and (\ref{2:eta-tk-eq}). 

The equations (\ref{5:xieta-tk-eq}) 
represent the two-component KP hierarchy 
\footnote{To be more precies, 
this is a restriction to the subspace 
$\bst + \bar{\bst} = \bszero$ of all flows 
$\bst = (t_1,t_2,\ldots)$ and 
$\bar{\bst} = (\bar{t}_1,\bar{t_2},\ldots)$ 
of the two-component KP hierarchy.}
in the language of the two-component 
Sato Grassmannian $\Gr^{(2)}$.  
Alongside these equations, $\xi$ and $\eta$ satisfy 
the constraints (\ref{4:xieta-constraint}) as well.  
This leads to a reduction of the two-component 
KP hierarchy. The reduced system is nothing but 
the the AKNS hierarchy.  This fact is commonly known 
in the literature (e.g., the paper of Kac and van de Leur 
\cite{KvdL03}). 

We can explain this link with the AKNS hierarchy 
in the perspective of the wave functions 
$u_i(z),\bar{u}_i(z)$ as well.  
It is straightforward to verify that $g$ and $\bar{g}$ 
satisfy the differential equations   
\beq
\begin{aligned}
  \frac{\rd g}{\rd t_k} 
  &= (a\otimes\Lambda^k - U(a\otimes\calO_k))g 
    - g(a\otimes\Lambda^k),\\
  \frac{\rd\bar{g}}{\rd t_k} 
  &= (a\otimes\Lambda^k - U(a\otimes\calO_k))\bar{g}
    - \bar{g}(a\otimes\Lambda^k) 
\end{aligned}
  \label{5:g-tk-eq}
\eeq
for $k = 1,2,\ldots$. 
Let $\bsu(z)$ and $\bar{\bsu}(z)$ be defined 
with the $2 \times 2$ plane wave factor 
\[
  \rho(z) = \exp\left(\sum_{k=1}^\infty t_kaz^k\right). 
\]
As a consequence of (\ref{5:g-tk-eq}), 
$\bsu(z)$ and $\bar{\bsu}(z)$ satisfy 
the differential equations 
\beq
\begin{aligned}
  \frac{\rd\bsu(z)}{\rd t_k} 
  &= \left(a\otimes\Lambda^k - U(a\otimes\calO_k)
    \right)\bsu(z),\\
  \frac{\rd\bar{\bsu}(z)}{\rd t_k} 
  &= \left(a\otimes\Lambda^k - U(a\otimes\calO_k)
    \right)\bar{\bsu}(z)
\end{aligned}
  \label{5:u(z)-tk-eq}
\eeq
for $k = 1,2,\ldots$.  
The $0$-th components of these equations give 
\[
\begin{aligned}
  \frac{\rd u_0(z)}{\rd t_k} 
  &= au_k(z) - \sum_{l=0}^{k-1}u_{0l}au_{k-l-1}(z),\\
  \frac{\rd \bar{u}_0(z)}{\rd t_k} 
  &= a\bar{u}_k(z) 
    - \sum_{l=0}^{k-1}u_{0l}a\bar{u}_{k-l-1}(z) 
\end{aligned}  
\]
Using (\ref{4:ui(z)-u0(z)}), one can rewrite 
the terms $u_k(z)$ and $u_{k-l-1}(z)$ 
on the right hand side. This leads to 
the differential equations 
\beq
\begin{aligned}
  \frac{\rd u_0(z)}{\rd t_k} &= A_k(z)u_0(z),\\
  \frac{\rd\bar{u}_0(z)}{\rd t_k} &= A_k(z)\bar{u}_0(z),
\end{aligned}
  \label{4:uubar0(z)-tk-eq}
\eeq
where $A_k(z)$ is a polynomials in $z$ of the form 
\[
  A_k(z) = az^k + a_{k1}z^{k-1} + \cdots + a_{kk} 
\]
with $2 \times 2$ matrix coefficients 
$a,a_{k1},\ldots,a_{kk}$.  
These differential equations can be identified with 
the auxiliary linear equations of the AKNS hierarchy.  

\begin{remark}
Since $u_0(z)$ and $\bar{u}_0(z)$ satisfy the same 
linear differential equations (\ref{4:uubar0(z)-tk-eq}), 
there is a $2 \times 2$ matrix-valued 
\textit{$\bst$-independent} function $c(z)$ 
such that 
\beq
  \bar{u}_0(z) = u_0(z)c(z). 
  \label{5:RH-eq}
\eeq
Consequently, the amplitude parts 
\[
  W(z) = p\bsc(z),~~
  \overline{W}(z) = \bar{p}\bsc(z), 
\]
of $u_0(z),\bar{u}_0(z)$ satisfy 
the relation 
\beq
  \overline{W}(z) = W(z)\rho(z)c(z)\rho(z)^{-1}. 
  \label{5:RH-eq2}
\eeq
In a suitable analytic setting \cite{ZS79}, 
$W(z)$ and $\overline{W}(z)$ thus form 
a \textit{Riemann-Hilbert pair}. 
\end{remark}

\begin{remark}
The differential equations (\ref{5:g-tk-eq}) 
for $g$ and $\bar{g}$ can be cast into 
the differential equations 
\[
  \frac{\rd\eta^{(2)}}{\rd t_k} 
  = \left(a\otimes\Lambda^k - U(a\otimes\calO_k)\right)\eta^{(2)} 
  - \eta^{(2)}\diag(a\otimes\Lambda^k,a\otimes\Lambda^k)  
\]
for the $\eta$-matrix 
\[
  \eta^{(2)} 
  = \left(\begin{array}{c|c}
    g & \bar{g}
    \end{array}\right) 
  = \left(\begin{array}{c|c|c|c}
    \delta_{ij}E_2 - u_{i,-j-1} & 0 & \delta_{ij}E_2 & u_{i,-j-1}\\
    \hline
    - u_{i,-j-1} & \delta_{ij}E_2 & 0 & \delta_{ij}E_2 + u_{i,-j-1}
    \end{array}\right).
\]
This extended $\eta$-matrix can be used for 
a Grassmannian description of the AKNS hierarchy 
with both positive and negative flows.  
\end{remark}

\subsection{$U$-system for ASDYM hierarchy}

In the following, we consider $\GL(r)$ gauge theory 
with $r > 1$.  Let $\bsx = (x_0,x_1,\ldots)$ and 
$\bsy = (y_0,y_1,\ldots)$ be the space-time variables 
of the ASDYM hierarchy.  
The first four variables $(x_0,x_1,y_0,y_1)$ 
correspond to (possibly complexified) 
4D space-time coordinates. 
The $U$-matrix now depend on these variables. 

\begin{remark}
One can use a single set $\bsx$ of variables  
labelled by positive and negative integers  
in which 4D space-time coordinates 
are identified with $(x_{n+1},x_m,x_{`m+1},x_n)$
\cite{LQYZ22,LQZ23,LZ25,LHHZ25}.  
This simplifies some aspects, but we dare 
to adopt the conventional formulation 
\cite{Takasaki90CMP} that employs 
the two sets $\bsx,\bsy$ of variables 
(see Appendix C).  This formulation 
takes into account the relationship 
with twistor theory \cite{MWbook96}.  
On the other hand, we do not consider 
negative flows to avoid technical and 
notational complexity. 
\end{remark}

Direct linearization of the ASDYM hierarchy 
can be achieved by the $U$-matrix satisfying 
the differential equations 
\beq
\begin{aligned}
  \frac{\rd U}{\rd x_k} 
  &= \left(E_r\otimes\Lambda - U(E_r\otimes\calO_1)
    \right)\frac{\rd U}{\rd x_{k-1}},\\
  \frac{\rd U}{\rd y_k} 
  &= \left(E_r\otimes\Lambda - U(E_r\otimes\calO_1)
     \right)\frac{\rd U}{\rd y_{k-1}}
\end{aligned}
  \label{5:U-xy-eq}
\eeq
for $k = 1,2,\ldots$ and the algebraic 
constraint (\ref{4:U-constraint}). 
These equations can be linearized by 
the transformation 
\[
  U = \left(E - U(E_r\otimes\Omega)\right)C 
\]
to the linear differential equations 
\beq
  \frac{\rd C}{\rd x_k} 
  = (E_r\otimes\Lambda)\frac{\rd C}{\rd x_{k-1}},~~
  \frac{\rd C}{\rd y_k} 
  = (E_r\otimes\Lambda)\frac{\rd C}{\rd y_{k-1}}. 
  \label{5:C-xy-eq}
\eeq
The constraint for $C$ takes the same form 
as (\ref{5:C-constraint}). 

Compared to the $U$-system for the KP and 
AKNS hierarchies, the $U$-system (\ref{5:U-xy-eq}) 
and its linearization (\ref{5:C-xy-eq}) 
have a somewhat exotic form.  In particular, 
derivatives of $U$ are contained in 
\textit{both hand sides} of the equations  
including the nonlinear term.  As we shall see 
in the subsequent consideration, 
this reflects twistor-theoretic features 
of the ASDYM hierarchy. 

The Grassmannian structure underlying 
this $U$-system, too, turns out to take 
an exotic form.  The equations (\ref{5:U-xy-eq}) 
for $U$ become the differential equations 
\beq
\begin{aligned}
  \frac{\rd u_{ij}}{\rd x_k}
  &= \frac{\rd u_{i+1,j}}{\rd x_{k-1}} 
  - u_{i0}\frac{\rd u_{0j}}{\rd x_{k-1}},\\
  \frac{\rd u_{ij}}{\rd y_k}
  &= \frac{\rd u_{i+1,j}}{\rd y_{k-1}} 
  - u_{i0}\frac{\rd u_{0j}}{\rd y_{k-1}}  
\end{aligned}
  \label{5:uij-xy-eq}
\eeq
for $u_{ij}$, $i,j \in \ZZ$.  Alongside 
these differential equations, $u_{ij}$'s 
satisfy the constraint (\ref{4:uij-constraint}) 
as well. Identifying $u_{ij}$'s for $i,j \ge 0$ 
with the affine coordinates $w_{ij}$ 
of $\Gr^{(r)}$, one can construct 
the matrices $\xi,\eta$ introduced in the last section.  
The differential equations (\ref{5:uij-xy-eq}) 
for these matrix elements can be cast 
into differential equations of the form 
\beq
  \frac{\rd\xi}{\rd x_k} 
  = (E_r\otimes\Lambda)\frac{\rd\xi}{\rd x_{k-1}} - \xi\calP_k,~~
  \frac{\rd\xi}{\rd y_k} 
  = (E_r\otimes\Lambda)\frac{\rd\xi}{\rd y_{k-1}} - \xi\calQ_k,
  \label{5:xi-xy-eq}
\eeq
\beq
  \frac{\rd\eta}{\rd x_k} 
  = \frac{\rd\eta}{\rd x_{k-1}}(E_r\otimes\Lambda) - \calR_k\eta,~~
  \frac{\rd\eta}{\rd y_k} 
  = \frac{\rd\eta}{\rd y_{k-1}}(E_r\otimes\Lambda) - \calS_k\eta 
  \label{5:eta-xy-eq}
\eeq
and the algebraic constraints (\ref{4:xieta-constraint}). 
The matrices $\calP_k,\calQ_k,\calR_k,\calS_k$ 
are defined as 
\[
\begin{aligned}
  \calP_k &= \left(\begin{array}{cl}
         0~&(i<-1)\\\hline
         \rd w_{0,-j-1}/\rd x_{k-1} & (i=-1)
         \end{array}\right)_{i,j\in\ZZ_{<0}},\\
  \calQ_k &= \left(\begin{array}{cl}
         0 & (i<-1)\\\hline
         \rd w_{0,-j-1}/\rd y_{k-1} & (i=-1)
         \end{array}\right)_{i,j\in\ZZ_{<0}},\\
  \calR_k &= \left(\begin{array}{c|c}
         - \dfrac{\rd w_{i,0}}{\rd x_{k-1}}~~(j=0) & 0~~(j>0)
         \end{array}\right)_{i,j\in\ZZ_{\ge 0}},\\
  \calS_k &= \left(\begin{array}{c|c}
         - \dfrac{\rd w_{i,0}}{\rd y_{k-1}}~~(j=0) & 0~~(j>0)
         \end{array}\right)_{i,j\in\ZZ_{\ge 0}}. 
\end{aligned}
\]
(\ref{5:xi-xy-eq}) and (\ref{5:eta-xy-eq}) agree 
with the equations arising in the the previously known 
Grassmannian formalism of the ASDYM hierarchy 
\cite{Takasaki84CMP}.

\subsection{Auxiliary linear equations of ASDYM hierarchy}

The generating function (\ref{4:W(z)-w0j}) of $w_{0j}$'s  
satisfies the well-known auxiliary linear equations 
\cite{Nakamura88,Takasaki90CMP,ACT93} (see Appendix C) 
\beq
\begin{aligned}
  \frac{\rd W(z)}{\rd x_k} 
  - z\frac{\rd W(z)}{\rd x_{k-1}} + A_kW(z) &= 0,\\
  \frac{\rd W(z)}{\rd y_k} 
  - z\frac{\rd W(z)}{\rd y_{k-1}} + B_kW(z) &= 0, 
\end{aligned}
  \label{5:W(z)-xy-eq}
\eeq
where 
\[
  A_k = - \frac{\rd w_{00}}{\rd x_{k-1}},~~
  B_k = - \frac{\rd w_{00}}{\rd y_{k-1}}. 
\]
Let us derive these auxiliary linear equations 
in the perspective of $u_i(z)$'s and $\bar{u}_i(z)$'s. 
This serves as a test for verifying 
the validity of the $U$-system (\ref{5:U-xy-eq}) 
as a direct linearization scheme of the ASDYM hierarchy.  

A clue to this issue is the following differential 
equations satisfied by $g$ and $\bar{g}$. 

\begin{prop}
$g$ and $\bar{g}$ satisfy the differential equations 
\beq
\begin{aligned}
  \frac{\rd g}{\rd x_k} 
  &= \left(E_r\otimes\Lambda - U(E_r\otimes\calO_1)
     \right)\frac{\rd g}{\rd x_{k-1}},\\
  \frac{\rd g}{\rd y_k}
  &= \left(E_r\otimes\Lambda - U(E_r\otimes\calO_1)
     \right)\frac{\rd g}{\rd y_{k-1}}
\end{aligned}
  \label{5:g-xy-eq}
\eeq
and 
\beq
\begin{aligned}
  \frac{\rd\bar{g}}{\rd x_k} 
  &= \left(E_r\otimes\Lambda - U(E_r\otimes\calO_1)
     \right)\frac{\rd\bar{g}}{\rd x_{k-1}},\\
  \frac{\rd\bar{g}}{\rd y_k}
  &= \left(E_r\otimes\Lambda - U(E_r\otimes\calO_1)
     \right)\frac{\rd\bar{g}}{\rd y_{k-1}}
\end{aligned}
  \label{5:gbar-xy-eq}
\eeq
for $k = 1,2,\ldots$. 
\end{prop}

\proof
Using the equations (\ref{5:U-xy-eq}) of $U$, 
we can express the $x_k$-derivatives of $g$ as 
\[
  \frac{\rd g}{\rd x_k}
  = - \frac{\rd U}{\rd x_k}(E_r\otimes\Omega) 
  = - \left(E_r\otimes\Lambda - U(E_r\otimes\calO_1)
       \right)\frac{\rd U}{\rd x_{k-1}}(E\otimes\Omega). 
\]
On the other hand, we have 
\[
  \left(E_r\otimes\Lambda - U(E_r\otimes\calO_1)
  \right)\frac{\rd g}{\rd x_{k-1}} 
  = - \left(E_r\otimes\Lambda - U(E_r\otimes\calO_1)
    \right)\frac{\rd U}{\rd x_{k-1}}(E_r\otimes\Omega). 
\]
Thus the first equation of (\ref{5:g-xy-eq}) 
turns out to hold.  The other equations can be verified 
in the same way. 
\qed

We now consider a $\GL(r)$-valued 
`plane wave' factor $\rho(z)$ that satisfies 
the differential equations 
\beq
  \frac{\rd\rho(z)}{\rd x_k} 
  - z\frac{\rd\rho(z)}{\rd x_{k-1}} = 0, ~~
  \frac{\rd\rho(z)}{\rd y_k} 
  - z\frac{\rd\rho(z)}{\rd y_{k-1}} = 0 
  \label{5:rho(z)-xy-eq}
\eeq
for $k = 1,2,\ldots$, and construct  
the wave functions 
\[
  \bsu(z) = g\bsc(z)\rho(z),~~
  \bar{\bsu}(z) = \bar{g}\bsc(z)\rho(z). 
\]

\begin{remark}
The differential equations  (\ref{5:rho(z)-xy-eq}) 
reflect twistor-theoretic properties 
\cite{MWbook96} of the ASDYM hierarchy. 
These equations imply that the matrix elements 
of $\rho(z)$ are functions of $z$ and 
the generating functions 
\[
  \bsx[z] = \sum_{k=0}^\infty x_kz^k, ~~
  \bsy[z] = \sum_{k=0}^\infty y_kz^k
\]
of $\bsx$ and $\bsy$.  Such functions are called 
\textit{twistor functions} in the context 
of twistor theory.  In fact, the triple 
$(z,\bsx[z],\bsy[z])$ shows up 
in the geometric correspondence between 
the (extended) space-time and 
the 3D twistor space $\CC\PP^3$.  
\end{remark}

\begin{prop}
$\bsu(z)$ and $\bar{\bsu}(z)$ satisfy 
the differential equations 
\beq
\begin{aligned}
  \frac{\rd\bsu(z)}{\rd x_k} 
  &= \left(E_r\otimes\Lambda - U(E_r\otimes\calO_1)
     \right)\frac{\rd\bsu(z)}{\rd x_{k-1}},\\
  \frac{\rd\bsu(z)}{\rd y_k} 
  &= \left(E_r\otimes\Lambda - U(E_r\otimes\calO_1)
     \right)\frac{\rd\bsu(z)}{\rd y_{k-1}}
\end{aligned}
  \label{5:u(z)-xy-eq}
\eeq
and
\beq
\begin{aligned}
  \frac{\rd\bar{\bsu}(z)}{\rd x_k} 
  &= \left(E_r\otimes\Lambda - U(E_r\otimes\calO_1)
     \right)\frac{\rd\bar{\bsu}(z)}{\rd x_{k-1}},\\
  \frac{\rd\bar{\bsu}(z)}{\rd y_k} 
  &= \left(E_r\otimes\Lambda - U(E_r\otimes\calO_1)
     \right)\frac{\rd\bar{\bsu}(z)}{\rd y_{k-1}}
\end{aligned}
  \label{5:ubar(z)-xy-eq}
\eeq
for $k = 1,2,\ldots$. 
\end{prop}

\proof
Differentiating $\bsu(z)$ by $x_k$ 
and using (\ref{5:g-xy-eq}) and (\ref{5:rho(z)-xy-eq}), 
we have 
\[
\begin{aligned}
  \frac{\rd\bsu(z)}{\rd x_k} 
  &= \frac{\rd g}{\rd x_k}\bsc(z)\rho(z) 
    + g\bsc(z)\frac{\rd\rho(z)}{\rd x_k} \\
  &= \left(E_r\otimes\Lambda - U(E_r\otimes\calO_1)
     \right)\frac{\rd g}{\rd x_{k-1}}\bsc(z)\rho(z) 
    + g\bsc(z)z\frac{\rd\rho(z)}{\rd x_{k-1}}\\
  &= \left(E_r\otimes\Lambda - U(E_r\otimes\calO_1)
     \right)\frac{\rd g}{\rd x_{k-1}}\bsc(z)\rho(z)
    + g(E_r\otimes\Lambda)\bsc(z)\frac{\rd\rho(z)}{\rd x_{k-1}}
\end{aligned}
\]
On the other hand, the $x_{k-1}$-derivative becomes 
\[
  \frac{\rd\bsu(z)}{\rd x_{k-1}} 
  = \frac{\rd g}{\rd x_{k-1}}\bsc(z)\rho(z) 
    + g\bsc(z)\frac{\rd\rho(z)}{\rd x_{k-1}}. 
\]
Therefore 
\[
\begin{aligned}
  &\frac{\rd\bsu(z)}{\rd x_k} 
  - \left(E_r\otimes\Lambda - U(E_r\otimes\calO_1)
     \right)\frac{\rd\bsu(z)}{\rd x_{k-1}} \\
  &= \left(g(E_r\otimes\Lambda) 
      - \left(E_r\otimes\Lambda - U(E_r\otimes\calO_1)\right)g
     \right)\bsc(z)\frac{\rd\bsu(z)}{\rd x_{k-1}}.
\end{aligned}
\]
Because of the constraint (\ref{4:g-constraint}), 
the right hand side of the last equation vanishes. 
Thus we obtain the first equation of (\ref{5:u(z)-xy-eq}). 
The other equations can be verified in the same way. 
\qed

\begin{theorem}
$u_0(z)$ and $\bar{u}_0(z)$ satisfy 
the differential equations 
\beq
\begin{aligned}
  \frac{\rd u_0(z)}{\rd x_k} 
  &= z\frac{\rd u_0(z)}{\rd x_{k-1}} 
     + \frac{\rd u_{00}}{\rd x_{k-1}}u_0(z),\\
  \frac{\rd u_0(z)}{\rd y_k} 
  &= z\frac{\rd u_0(z)}{\rd y_{k-1}} 
     + \frac{\rd u_{00}}{\rd y_{k-1}}u_0(z)
\end{aligned}
  \label{5:u0(z)-xy-eq}
\eeq
and
\beq
\begin{aligned}
  \frac{\rd\bar{u}_0(z)}{\rd x_k} 
  &= z\frac{\rd\bar{u}_0(z)}{\rd x_{k-1}} 
     + \frac{\rd u_{00}}{\rd x_{k-1}}\bar{u}_0(z),\\
  \frac{\rd\bar{u}_0(z)}{\rd y_k} 
  &= z\frac{\rd\bar{u}_0(z)}{\rd y_{k-1}} 
     + \frac{\rd u_{00}}{\rd y_{k-1}}\bar{u}_0(z) 
\end{aligned}
  \label{5:ubar0(z)-xy-eq}
\eeq
for $k = 1,2,\ldots$.  
\end{theorem}

\proof
We can extract the differential equations 
\[
\begin{aligned}
  \frac{\rd u_0(z)}{\rd x_k} 
  &= \frac{\rd u_1(z)}{\rd x_{k-1}} 
     - u_{00}\frac{\rd u_0(z)}{\rd x_{k-1}},\\
  \frac{\rd u_0(z)}{\rd y_k} 
  &= \frac{\rd u_1(z)}{\rd y_{k-1}} 
     - u_{00}\frac{\rd u_0(z)}{\rd y_{k-1}}
\end{aligned}
\]
for $u_0(z)$ and
\[
\begin{aligned}
  \frac{\rd\bar{u}_0(z)}{\rd x_k} 
  &= \frac{\rd\bar{u}_1(z)}{\rd x_{k-1}} 
     - u_{00}\frac{\rd\bar{u}_0(z)}{\rd x_{k-1}},\\
  \frac{\rd\bar{u}_0(z)}{\rd y_k} 
  &= \frac{\rd\bar{u}_1(z)}{\rd y_{k-1}} 
     - u_{00}\frac{\rd\bar{u}_0(z)}{\rd y_{k-1}}
\end{aligned}
\]
for $\bar{u}_0(z)$ from the $0$-th component 
of (\ref{5:u(z)-xy-eq}) and (\ref{5:ubar(z)-xy-eq}). 
As a consequence of the relations 
(\ref{4:ui(z)-u0(z)}) and (\ref{4:ubari(z)-ubar0(z)}), 
$u_1(z)$ and $\bar{u}_1(z)$ can be expressed as 
\[
  u_1(z) = (E_rz + u_{00})u_0(z),~~ 
  \bar{u}_1(z) = (E_rz + u_{00})\bar{u}_0(z).
\]
Substituting this expression, we obtain 
(\ref{5:u0(z)-xy-eq}) and (\ref{5:ubar0(z)-xy-eq}). 
\qed

\begin{cor}
The amplitude part 
\[
  W(z) = p\bsc(z),~~
  \bar{W}(z) = \bar{p}\bsc(z). 
\]
of $u_0(z),\bar{u}_0(z)$ satisfy 
the auxiliary linear equations 
(\ref{5:W(z)-xy-eq}). 
\end{cor}

\proof 
$u_0(z)$ and $\bar{u}_0(z)$takes the factorized form 
\[
  u_0(z) = W(z)\rho(z),~~
  \bar{u}_0(z) = \overline{W}(z)\rho(z). 
\]
Thanks to the equations (\ref{5:rho(z)-xy-eq}), 
$\rho(z)$ can be factored out from the equations 
(\ref{5:u0(z)-xy-eq}) and (\ref{5:ubar0(z)-xy-eq}) 
satisfied by $u_0(z),\bar{u}_0(z)$. 
\qed

\begin{remark}
This result indicates that the presence of $\rho(z)$ 
is superfluous for the auxiliary linear equations 
of the ASDYM hierarchy.  This factor, 
however, plays a role in some cases 
where solutions of the ASDYM hierarchy 
are obtained as a variant of solutions of 
the AKNS hierarchy \cite{LQYZ22,LQZ23,LZ25,Takasaki83}. 
\end{remark}

\begin{remark}
Since $u_0(z)$ and $\bar{u}_0(z)$ satisfy 
the same linear differential equations, 
an equality of the form 
\[
  \bar{u}_0(z) = u_0(z)c(z), 
\]
analogous to the relation (\ref{5:RH-eq}), 
also holds in the present case. 
The $r \times r$ matrix $c(z)$, however, 
may depend on $\bsx$ and $\bsy$ as far as 
they satisfy the differential equations 
\[
  \frac{\rd c(z)}{\rd x_k} = z\frac{\rd c(z)}{\rd x_{k-1}},~~
  \frac{\rd c(z)}{\rd y_k} = z\frac{\rd c(z)}{\rd y_{k-1}} 
\]
for $k = 1,2,\ldots$.  In other words, 
the matrix elements of $c(z)$ are twistor functions 
in the aforementioned sense. 
They are nothing but the twistorial data 
of the ASDYM fields \cite{MWbook96}.  
One can rewrite the relation between 
$u_0(z)$ and $\bar{u}_0(z)$ in terms of 
$W(z)$ and $\overline{W}(z)$ as 
\[
  \overline{W}(z) = W(z)\rho(z)c(z)\rho(z)^{-1}. 
\]
This is an analogue of (\ref{5:RH-eq2}) 
in the present case. 
\end{remark}

\subsection{Special solutions and Cauchy matrix}

Fu and Nijhoff's construction of special solutions 
to the KP hierarchy \cite{FN17,FN18,FuThesis}
can be generalized to the AKNS and ASDYM hierarchies. 
In appearance, $u_{ij}$'s take the same form 
as the case of the KP hierarchy: 
\beq
  u_{ij} 
  = \tp{\bsr}K^i(E_N + A\calM)^{-1}AL^j\bss 
  = \tp{\bsr}K^iA(E_M + \calM A)^{-1}L^j\bss.  
  \label{5:CM-uij}
\eeq
The constant matrices $K$ ($N \times N$ matrix), 
$L$ ($M \times M$ matrix) and $A$ ($N \times M$ matrix) 
are exactly the same as in the case (\ref{2:CM-uij}) 
of the KP hierarchy except that 
they are required to satisfy the algebraic relation 
\beq
  KA = AL. 
  \label{5:KA=AL}
\eeq
As demonstrated below, this condition 
ensures that $U$ satisfies the constraint 
(\ref{4:U-constraint}). 
The other building blocks 
$\bsr$ ($N\times r$ matrix), 
$\bss$ ($M\times r$ matrix) 
and $\calM$ ($M \times N$ matrix) 
are system-dependent. 
$\calM$ is the Cauchy matrix 
\[
  \calM = \left(\frac{\bss_i\tp{\bsr_j}}
          {\kappa_j - \lambda_i}\right)_{1\le i\le M,1\le j\le N}, 
\]
where $\bsr_i$ and $\bss_i$ denote 
the $i$-th row of $\bsr$ and $\bss$, 
and satisfies the Sylvester equation
\beq
  \calM K - L\calM = \bss\tp{\bsr}. 
  \label{5:Sylvester-eq}
\eeq

\begin{prop}
If the condition (\ref{5:KA=AL}) is fulfilled, 
the constraints (\ref{4:uij-constraint}) 
for $u_{ij}$'s are satisfied.   
\end{prop}

\proof
$u_{i+1,j} - u_{i,j+1}$ can be expressed as 
\[
  u_{i+1,j} - u_{i,j+1} 
  = \tp{\bsr}K^i\left(KA(E_M + \calM A)^{-1}  
        - (E_N + A\calM)^{-1}AL\right)K^j\bss 
\]
The matrix in the middle can be rewritten as 
\[
\begin{aligned}
  &KA(E_M + \calM A)^{-1} - (E_N + A\calM)^{-1}AL \\
  &= (E_N + A\calM)^{-1}\left((E_N + A\calM)KA 
      - AL(E_M + \calM A)\right)(E_M + \calM A)^{-1} \\
  &= (E_N + A\calM)^{-1}(KA - AL + A\calM KA - AL\calM A)
     (E_M + \calM A)^{-1}\\
  &= (E_N + A\calM)^{-1}A(\calM K - L\calM)A(E_M + \calM A)^{-1} \\
  &=  (E_N + A\calM)^{-1}A\bss\tp{\bsr}(E_N + A\calM)^{-1}A. 
\end{aligned}
\]
The relation (\ref{5:KA=AL}) among $A,K,L$ and 
the Sylvester equation (\ref{5:Sylvester-eq}) 
have been used in the last three lines.  
Consequently,   
\[
  u_{i+1,j} - u_{i,j+1} 
  = \tp{\bsr}K^i(E_N + A\calM)^{-1}A\bss\tp{\bsr}
    (E_N + A\calM)^{-1}AL^j\bss 
  = u_{i0}u_{0j}. 
\]
\qed

\begin{remark}
If $N = M$ and $A$ is an invertible matrix, 
we can rewrite the foregoing expression 
of $u_{ij}$ as 
\beq
  u_{ij} = \tp{\bsr}K^i(B + \calM)^{-1}L^j\bss, 
  \label{5:CM-uij2}
\eeq
where $B$ denotes $A^{-1}$ and obeys 
the algebraic relation 
\beq
  BK = LB. 
  \label{5:BK=LB}
\eeq
This case is known in the literature 
\cite{Zhao18,LQYZ22,LQZ23}. 
\end{remark}

\subsubsection{AKNS hierarchy}

This case is thoroughly studied in the literature 
in the context of the Cauchy matrix approach 
\cite{Zhao18,LQYZ22}.  
The matrices $\bsr$ ($N\times 2$ matrix) 
and $\bss$ ($M \times 2$ matrix) are required 
to satisfy the differential equations
\beq
  \frac{\rd\bsr}{\rd t_k} = K^k\bsr a,~~
  \frac{\rd\bss}{\rd t_k} = - L^k\bss a
  \label{5:rs-tpmk-eq}
\eeq
for $k = 1,2,\ldots$.  
The Cauchy matrix $\calM$ is constructed 
from these matrices and satisfies 
the differential equations 
\beq
  \frac{\rd\calM}{\rd t_k} 
  = \sum_{l=0}^{k-1}L^l\bss a\tp{\bsr}K^{k-l-1}
  \label{5:calM-tk-eq}
\eeq
for $k = 1,2,\ldots$.  Using these equations, 
one can confirm that $u_{ij}$'s of (\ref{5:CM-uij}) 
satisfy the evolution equations 
(\ref{5:uij-tk-eq}) in the same way 
as the case of the KP hierarchy. 

The wave functions $u_i(z)$ and $\bar{u}_i(z)$ 
obtained from (\ref{5:CM-uij}) become matrix-valued 
functions of the same form:
\[
\begin{aligned}
  u_i(z) 
  &= \left(z^i - \tp{\bsr}K^i(E_N + A\calM)^{-1}A 
          (zE_M - L)^{-1}\bss\right)\rho(z),\\
  \bar{u}_i(z) 
  &= \left(z^i - \tp{\bsr}K^i(E_N + A\calM)^{-1}A 
          (zE_M - L)^{-1}\bss\right)\rho(z). 
\end{aligned}
\]
This means that $u_i(z)$ and $\bar{u}_i(z)$ 
are local expressions of a globally defined 
matrix-valued function.  In particular, 
the connection matrix $c(z)$ in (\ref{5:RH-eq}) 
is equal to the unit matrix, implying 
that this is a degenerate case 
of the Riemann-Hilbert problem \cite{ZS79}.

\subsubsection{ASDYM hierarchy} 

In the work of Li et al. \cite{LQYZ22,LQZ23}, 
the foregoing special solution of the AKNS hierarchy 
and its higher rank analogues are positioned 
to be a solution of the ASDYM hierarchy. 
An origin of this coincidence lies in the fact 
that the plane wave factor 
\[
  \rho(z) = \exp\left(\sum_{k=1}^\infty t_kaz^k 
    + \sum_{k=1}^\infty t_{-k}az^{-k}\right) 
\]
of the AKNS hierarchy (with negative flows) 
turns into a twistor function satisfying 
the differential equations (\ref{5:rho(z)-xy-eq}) 
by replacing $t_{\pm k}$ with, e.g., $x_{k-1}$ and $y_{k-1}$.  

We have the following more general result. 
Ohta reported a similar result for Gram-type 
determinant solutions of the ASDYM equation 
employing the method of bilinearization \cite{Ohta24}. 

\begin{theorem}
\label{5:theorem-ASDYM-rs}
$u_{ij}$'s of (\ref{5:CM-uij}) become a solution 
of (\ref{5:uij-xy-eq}) if $\bsr$ and $\bss$ 
satisfy the differential equations 
\beq
\begin{aligned}
  \frac{\rd\bsr}{\rd x_k} = K\frac{\rd\bsr}{\rd x_{k-1}},&~~
  \frac{\rd\bsr}{\rd y_k} = K\frac{\rd\bsr}{\rd y_{k-1}},\\
  \frac{\rd\bss}{\rd x_k} = L\frac{\rd\bss}{\rd x_{k-1}},&~~
  \frac{\rd\bss}{\rd y_k} = L\frac{\rd\bss}{\rd y_{k-1}}.
\end{aligned}
  \label{5:rs-xy-eq}
\eeq
\end{theorem}

The equations (\ref{5:rs-xy-eq}) for $\bsr,\bss$ 
originate in the subsequent work of Li et al. \cite{LHHZ25} 
on soliton solutions of the ASDYM equation.  
They use these equations along with 
the Cauchy matrix to construct the soliton solutions 
by the method of Darboux transformations.  
The following lemma, which we shall use below, 
plays a key role in their approach.  

\begin{lemma}[\cite{LHHZ25}]
If $\bsr$ and $\bss$ are a solution 
of (\ref{5:rs-xy-eq}), the Cauchy matrix $\calM$ 
satisfies the differential equations 
\beq
\begin{aligned}
  \frac{\rd\calM}{\rd x_k} &= L\frac{\rd\calM}{\rd x_{k-1}} 
    + \bss\frac{\rd\tp{\bsr}}{\rd x_{k-1}},\\
  \frac{\rd\calM}{\rd y_k} &= L\frac{\rd\calM}{\rd y_{k-1}} 
    + \bss\frac{\rd\tp{\bsr}}{\rd y_{k-1}} 
\end{aligned}
  \label{5:calM-xy-eq}
\eeq
for $k = 1,2,\ldots$. 
\end{lemma}

\begin{remark}
By the Sylvester equation (\ref{5:Sylvester-eq}), 
one can rewrite (\ref{5:calM-xy-eq}) as 
\[
\begin{aligned}
  \frac{\rd\calM}{\rd x_k} &= \frac{\rd\calM}{\rd x_{k-1}}K  
    - \frac{\rd\bss}{\rd x_{k-1}}\tp{\bsr},\\
  \frac{\rd\calM}{\rd y_k} &= \frac{\rd\calM}{\rd y_{k-1}}K  
    - \frac{\rd\bss}{\rd y_{k-1}}\tp{\bsr}. 
\end{aligned}
\]
\end{remark}

\proof
The equations (\ref{5:rs-xy-eq}) for $\bsr,\bss$ 
can be decomposed into the equations 
\beq
\begin{aligned}
  \frac{\rd\bsr_i}{\rd x_k} 
  = \kappa_i\frac{\rd\bsr_i}{\rd x_{k-1}},&~~
  \frac{\rd\bsr_i}{\rd y_k} 
  = \kappa_i\frac{\rd\bsr_i}{\rd y_{k-1}},\\
  \frac{\rd\bss_i}{\rd x_k} 
  = \lambda_i\frac{\rd\bss_i}{\rd x_{k-1}},&~~
  \frac{\rd\bss_i}{\rd y_k} 
  = \lambda_i\frac{\rd\bss_i}{\rd y_{k-1}} 
\end{aligned}
  \label{5:rs-xy-eq2}
\eeq
for the row vectors $\bsr_i,\bss_i$ of $\bsr,\bss$. 
Differentiating $\calM$ by $x_k$ and $x_{k-1}$ yields 
\[
\begin{aligned}
  \frac{\rd\calM}{\rd x_k} 
  &= \left(\frac{1}{\kappa_j - \lambda_i}
       \left(\frac{\rd\bss_i}{\rd x_k}\tp{\bsr_j} 
       + \bss_i\frac{\rd\tp{\bsr_j}}{\rd x_k}\right)\right),\\
  \frac{\rd\calM}{\rd x_{k-1}} 
  &= \left(\frac{1}{\kappa_j - \lambda_i}
       \left(\frac{\rd\bss_i}{\rd x_{k-1}}\tp{\bsr_j} 
       + \bss_i\frac{\rd\tp{\bsr_j}}{\rd x_{k-1}}\right)\right), 
\end{aligned}
\]
hence 
\[
  L\frac{\rd\calM}{\rd x_{k-1}} 
  = \left(\frac{\lambda_i}{\kappa_j - \lambda_i}
      \left(\frac{\rd\bss_i}{\rd x_{k-1}}\tp{\bsr_j} 
      + \bss_i\frac{\rd\tp{\bsr_j}}{\rd x_{k-1}}\right)\right). 
\]
Subtracting the third equation 
from the first equation and 
using the equations (\ref{5:rs-xy-eq}), 
we have  
\[
  \frac{\rd\calM}{\rd x_k} - L\frac{\rd\calM}{\rd x_{k-1}} 
  = \left(\bss_i\frac{\rd\tp{\bsr_j}}{\rd x_{k-1}}\right) 
  = \bss\frac{\rd\tp{\bsr}}{\rd x_{k-1}}.
\]
This is the first equation of (\ref{5:calM-xy-eq}). 
The second equation can be derived in the same way. 
\qed

Let us proceed to the proof of Theorem 
\ref{5:theorem-ASDYM-rs}. 

\proof
The $x_k$-derivative of $u_{ij}$ can be written as 
\[
\begin{aligned}
  \frac{\rd u_{ij}}{\rd x_k}
  &=  - \tp{\bsr}K^i(E_N+A\calM)^{-1}
       A\frac{\rd\calM}{\rd x_k} 
       (E_N+A\calM)^{-1}L^j\bss  \\
  &~~~~\mbox{} +\frac{\rd\tp{\bsr}}{\rd x_k}
       K^i(E_N+A\calM)^{-1}AL^j\bss \\
  &~~~~\mbox{} + \tp{\bsr}K^i(E_N+A\calM)^{-1}AL^j
       \frac{\rd\bss}{\rd x_k}. 
\end{aligned}
\]
The $x_{k-1}$-derivative of $u_{i+1,j}$ can be 
written in a similar form.  Subtracting the latter 
from the former and using (\ref{5:rs-xy-eq}) and 
(\ref{5:calM-xy-eq}), we have 
\[
  \frac{\rd u_{ij}}{\rd x_k} - \frac{\rd u_{i+1,j}}{\rd x_{k-1}}
  = \tp{\bsr}K^i(\heartsuit + \spadesuit + \clubsuit), 
\]
where 
\[
\begin{aligned}
  \heartsuit
  &= - \left((E_N+A\calM)^{-1}AL - K(E_N+A\calM)^{-1}A
       \right)\frac{\rd\calM}{\rd x_{k-1}}
       (E_N+A\calM)^{-1}AL^j\bss,\\
  \spadesuit
  &= - (E_N+A\calM)^{-1}A\bss\frac{\rd\tp{\bsr}}{\rd x_{k-1}}
       (E_N+A\calM)^{-1}AL^j\bss, \\
  \clubsuit
  &= \left((E_N+A\calM)^{-1}AL - K(E_N+A\calM)^{-1}A
     \right)L^j\frac{\rd\bss}{\rd x_{k-1}}. 
\end{aligned}
\]
We can rewrite 
$(E_N+A\calM)^{-1}AL - K(E_N+A\calM)^{-1}A$ as 
\[
\begin{aligned}
  &(E_N+A\calM)^{-1}AL - K(E_N+A\calM)^{-1}A \\
  &= (E_N+A\calM)^{-1}AL - KA(E_M+\calM A)^{-1}\\
  &=  (E_N+A\calM)^{-1}\left(AL(E_M+\calM A) 
        - (E_N+A\calM)KA\right)(E_M+\calM A)^{-1}\\
  &= (E_N+A\calM)^{-1}\left(AL -KA 
        + A(L\calM - \calM K)\right)(E_M+\calM A)^{-1}\\
  &= - (E_N+A\calM)^{-1}A\bss\tp{\bsr}A(E_M+\calM A)^{-1}\\
  &= - (E_N+A\calM)^{-1}A\bss\tp{\bsr}(E_N+A\calM)^{-1}A  
\end{aligned}
\]
using the condition (\ref{5:KA=AL}) 
and the Sylvester equation (\ref{5:Sylvester-eq}) 
in the last part.  Consequently, 
\[
  \tp{\bsr}K^i(\heartsuit + \spadesuit + \clubsuit ) 
  = - \tp{\bsr}K^i(E_N+A\calM)^{-1}A\bss\diamondsuit 
  = - u_{i0}\diamondsuit,  
\]
where
\[
\begin{aligned}
  \diamondsuit
  &= - \tp{\bsr}(E_N+A\calM)^{-1}A
      \frac{\rd\calM}{\rd x_{k-1}}(E_N+A\calM)^{-1}AL^j\bss \\
  &~~~~\mbox{} 
    + \frac{\rd\tp{\bsr}}{\rd x_{k-1}}(E_N+A\calM)^{-1}AL^j\bss\\
  &~~~~\mbox{}
    + \tp{\bsr}(E_N+A\calM)^{-1}AL^j\frac{\rd\bss}{\rd x_{k-1}} \\
  &= \frac{\rd u_{0j}}{\rd x_{k-1}}. 
\end{aligned}
\]
Thus we obtain the first equation of 
(\ref{5:uij-xy-eq}).  The second equation 
of (\ref{5:uij-xy-eq}) can be derived 
in the same way. 
\qed

\begin{remark}
If $\bsrho_i(z)$ and $\bssigma_i(z)$ are 
$r$-dimensional vectors of twistor functions, 
their special values 
\[
  \bsr_i = \bsrho_i(\kappa_i), ~~
  \bss_i = \bssigma_i(\lambda_i) 
\]
satisfy the equations (\ref{5:rs-xy-eq2}). 
Solutions of the equations (\ref{5:rs-xy-eq}) 
can be thus obtained from twistor functions.  
In many papers \cite{LQYZ22,LQZ23,LZ25,LHHZ25}, 
the twistor functions are (implicitly) 
chosen to be exponential functions like 
the plane wave factor of the AKNS hierarchy. 
In contrast, Ohta constructed Gram-type 
determinant solutions of the ASDYM equation 
from rational twistor functions \cite{Ohta24}. 
\end{remark}

\section{Conclusion}

It seems that Fu and Nijhoff developed 
their infinite matrix formalism of 
direct linearization to avoid some unnatural 
aspects of the conventional formulation 
of the KP-type systems. As one of such issues, 
they specifically point out the use of 
pseudo-differential operators 
in a spatial coordinate $x$ \cite{FN18}. 
Certainly, choosing $t_1$ among the independent 
variables $t_1,t_2,\ldots$ 
as the spatial variable $x$ 
breaks symmetry of the system.  This issue 
is resolved, e.g., in Hirota's bilinear formalism.  
Sato's approach is also a solution to the same issue, 
because all $t_k$'s are treated on an equal footing  
in the dynamical system on the Grassmannian. 
Conceptually, the Grassmannian description 
is much closer to Fu and Nijhoff's 
direct linearization scheme because 
Sato's approach also employs an infinite matrix 
$\xi$ (and its dual $\eta$) to represent 
a point of the Grassmannian. 

In this paper, we have uncovered integrable 
and geometric structures hidden 
in Fu and Nijhoff's $U$-system. A clue is the fact 
that the KP hierarchy, written in terms of 
the affine coordinates $w_{ij}$ 
of the Sato Grassmannian, is a subsystem 
of Fu and Nijhoff's $U$-system.  
An extension of the $U$-system by negative flows 
contains another copy of the KP hierarchy 
describing the negative flows.  The extended 
$U$-system itself turns out to be equivalent 
to the two-component KP hierarchy, 
and the underlying geometric structure 
is the two-component Sato Grassmannian. 
The vector-valued wave functions $\bsu(z),\bar{\bsu}(z)$ 
and their dressing matrices $g,\bar{g}$ 
are used as very efficient tools to elucidate 
these integrable and geometric structures.  

We have generalized this perspective  
to the AKNS and ASDYM hierarchies.  To this end, 
we have introduced a multi-component version 
of the $U$-system and imposed a constraint.  
The constrained two-component $U$-system turns out 
to give a reformulation of the AKNS hierarchy. 
It will be obvious that this construction 
can be generalized to a higher rank version 
of the AKNS hierarchy.  Moreover, 
although omitted due to a lack of spaces, 
negative flows of the these hierarchies 
can be treated in the same way as in the case 
of Fu and Nijhoff's $U$-system. 
On the other hand, 
modifying the evolution equations of $U$ 
into an exotic form, we have been able 
to derive the ASDYM hierarchy.  
This modification is motivated by the well-known 
auxiliary linear equations of the ASDYM hierarchy. 
We have also examined special solutions obtained 
by the Cauchy matrix approach in terms 
of the multi-component constrained $U$-matrix. 

Fu and Nijhoff's infinite matrix formalism 
of direct linearization deserves 
to be studied further.  An intriguing issue 
is the treatment of the BKP and CKP hierarchies. 
Fu and Nijhoff presented a formulation 
of the relevant $U$-system \cite{FN17,FN18}. 
Presumably the orthogonal/symplectic Grassmannian 
will show up there as an underlying geometric 
structure.  The orthogonal Grassmannian structure 
of the BKP hierarchy and its relatives is usually 
understood in the language of a fermion Fock space 
\cite{DJKM82,KvdL98}.  Fu and Nijhoff's 
infinite matrix formalism might provide 
another framework to accommodate 
the orthogonal/symplectic Grassmannian.

\subsection*{Acknowledgements}

This work is partly supported by the JSPS Kakenhi Grant 
JP24K06724.  The author is grateful to Masashi Hamanaka, 
Shangshuai Li, Yasuhiro Ohta and Saburo Kahei 
for valuable comments.

\appendix

\section{KP hierarchy and affine coordinates 
of Sato Grassmannian}

The dressing operator of the KP hierarchy 
\cite{SS82,Sato89} is a pseudo-differential operator 
of the form 
\[
  W = 1 + \sum_{n=1}^\infty w_n\rd_x^{-n}, 
  ~~ \rd_x = \rd/\rd t_1, 
\]
and satisfies the Sato-Wilson equations 
\beq
  \frac{\rd W}{\rd t_k} 
  = \left(W\rd_x^kW^{-1}\right)_{\ge 0}W - W\rd_x^k 
  = - \left(W\rd_x^kW^{-1}\right)_{<0}W 
  \label{A:SW-eq}
\eeq
for $k = 1,2,\ldots$. 
$(~~)_{\ge 0}$ and $(~~)_{<0}$ stand for the projection 
to the nonnegative/negative power part 
of a pseudo-differential operator 
\[
  \left(\sum_{n=-\infty}^\infty a_n\rd_x^n\right)_{\ge 0} 
  = \sum_{n\ge 0}a_n\rd_x^n,~~
  \left(\sum_{n=-\infty}^\infty a_n\rd_x^n\right)_{<0} 
  = \sum_{n<0}a_n\rd_x^n. 
\]
The associated Lax operator $L$ is obtained 
from $W$ as 
\[
  L = W\rd_xW^{-1}
\]
and satisfies the Lax equations 
\[
  \frac{\rd L}{\rd t_k} = [B_k,L],~~
  B_k = \left(W\rd_x^kW^{-1}\right)_{\ge 0},~~
  k = 1,2,\ldots. 
\]

The dressing operator $W$ is used to express 
the wave function 
\[
  \Psi(z) = \left(1 + \sum_{n=1}^\infty w_nz^{-n}\right)\rho(z) 
\]
as the dressed wave function 
\[
  \Psi(z) = W\rho(z) 
\]
of the undressed wave function 
\[
  \rho(z) = \exp\left(\sum_{k=1}^\infty t_kz^k\right). 
\]
The integral powers of $\rd_x^n$, $n \in \ZZ$, 
are understood to act on $\rho(z)$ as 
\[
  \rd_x^n \rho(z) = z^n\rho(z). 
\]
This is consistent with the non-commutative 
ring structure of the formal pseudo-differential 
operators. In particular, the Lax operator $L$ 
acts on $\Psi(z)$ as 
\[
  L\Psi(z) = z\Psi(z). 
\]
As a consequence of the Sato-Wilson equation 
(\ref{A:SW-eq}),$\Psi(z)$ satisfies 
the auxiliary linear equations 
\[
  \frac{\rd\Psi(z)}{\rd t_k} = B_k\Psi(z), 
  ~~ k = 1,2,\ldots. 
\]

The affine coordinates of the top cell 
of the Sato Grassmannian can be read 
from the higher-order dressing operators 
\cite{Takasaki89,Takasaki89RMP}   
\[
  W_i = (\rd_x^i W^{-1})_{\ge 0}W, 
  ~~ i = 1,2,\ldots. 
\]
As one can see by rewriting these operators as 
\[
  W_i = \rd_x^i - (\rd_x^i W^{-1})_{<0}W,
\]
$W_i$'s are pseudo-differential operators 
of the form 
\beq
  W_i = \rd_x^i - \sum_{j=0}^\infty w_{ij}\rd_x^{-j-1}. 
  \label{A:Wi-wij-rel}
\eeq
The coefficient $w_{ij}$ showing up here 
are exactly the affine coordinates.  
The coefficients $w_n$ of $W$ can be recovered as 
\[
  w_n = - w_{0,n-1}. 
\]
\begingroup
These higher-order dressing operators $W_i$ 
originate in a $\calD$-module structure 
introduced by Sato \cite{Sato89,Takasaki89}. 

The fact that $w_{ij}$'s are affine coordinates 
of the Sato Grassmannian can be seen from 
the evolution equations satisfied 
by these coefficients \cite{Takasaki89,Takasaki89RMP}. 
The higher-order dressing operators $W_i$ 
satisfy the algebraic relations 
\beq
  \rd_x W_i = W_{i+1} - w_{i0}W_0 
  \label{A:rd-Wi-eq}
\eeq
and the evolution equations 
\beq
  \frac{\rd W_i}{\rd t_k} 
  = W_{i+k} - \sum_{l=0}^{k-1}w_{il}W_{k-l-1} 
    - W_i\rd_x^k, 
  ~~ k = 1,2,\ldots. 
  \label{A:Wi-tk-eq}
\eeq
One can rewrite (\ref{A:rd-Wi-eq}) as 
\[
  \frac{\rd W_i}{\rd x} = W_{i+1} - w_{i0}W_0 - W_i\rd_x, 
\]
which is a special ($k = 1$) case 
of (\ref{A:Wi-tk-eq}).  (\ref{A:Wi-tk-eq}) 
is a consequence of the Sato-Wilson equations 
(\ref{A:SW-eq}). These operator equations 
imply the evolution equations 
\[
  \frac{\rd w_{ij}}{\rd t_k} 
  = w_{i+k,j} - w_{i,j+k} - \sum_{l=0}^{k-1}w_{il}w_{k-l-1,j}, 
  ~~ i,j \ge 0, 
\]
for the coefficients $w_{ij}$.  Those equations 
can be cast into the evolution equations 
(\ref{2:xi-tk-eq}) and (\ref{2:eta-tk-eq}) 
of the matrices $\xi,\eta$ representing 
a moving point of the top cell of the Sato Grassmannian.

\section{2D Toda hierarchy and Sato Grassmannian}

The Grassmannian structure of the 2D Toda hierarchy 
\cite{Takasaki90LMP}, too, is described 
by dressing operators.  In the 2D Toda hierarchy 
\cite{UT84,Takasaki18}, there are two dressing operators 
of the form 
\[
  \calW = 1 + \sum_{n=1}^\infty w_n(s)e^{-n\rd_s},~~
  \overline{\calW} = \sum_{n=0}^\infty\bar{w}_n(s)e^{n\rd_s},~~
  ~~\rd_s = \rd/\rd s. 
\]
They are, so to speak, 
\textit{pseudo-difference operators} 
on the 1D lattice $\ZZ$ with coordinate $s$. 
There is a one-to-one correspondence 
between a pseudo-difference operator 
and a $\ZZ\times\ZZ$ matrix as 
\[
  \sum_{n=-\infty}^\infty a_n(s)e^{n\rd_s} 
  ~\longleftrightarrow~ 
  \sum_{n=-\infty}^\infty\diag(a(i))_{i\in \ZZ}\Lambda^n.
\]
In the following, we mostly use the matrix 
representation.  By this correspondence, 
$\calW$ and $\overline{\calW}$ can be identified 
with the $\ZZ\times\ZZ$ matrices 
\[
  \calW = \left(w_{j-i}(i)\right)_{i,j\in\ZZ},~~
  \overline{\calW} = \left(\bar{w}_{i-j}(i)\right)_{i,j\in\ZZ}, 
\]
where $w_n(i)$ and $\bar{w}_n(i)$ for $n < 0$ 
are understood to be equal to $0$. 
$\calW$ is lower-triangular, the diagonal 
elements being equal to $1$, and $\overline{\calW}$ 
is upper-triangular.  It is also assumed 
that $\bar{w}_0(i) \not= 0$, so that 
both $\calW$ and $\overline{\calW}$ are invertible. 

These dressing operators, viewed as 
$\ZZ\times\ZZ$-matrices, satisfy 
the Sato-Wilson equations 
\beq
\begin{aligned}
  \frac{\rd\calW}{\rd t_k} 
  &= \calB_k\calW - \calW\Lambda^k,&
  \frac{\rd\overline{\calW}}{\rd t_k} 
  &= \calB_k\overline{\calW},\\
  \frac{\rd\calW}{\rd\bar{t}_k}
  &= \overline{\calB}_k\calW,&
  \frac{\rd\overline{\calW}}{\rd\bar{t}_k} 
  &= \overline{\calB}_k\overline{\calW} 
   - \overline{\calW}\Lambda^{-k} 
\end{aligned}
  \label{B:SW-eq}
\eeq
for $k = 1,2,\ldots$.  
$\calB_k$ and $\bar{\calB}_k$ are matrices 
of the form 
\[
  \calB_k 
    = \sum_{l=0}^k\diag(b_{kl}(i))_{i\in\ZZ}\Lambda^i,~~
  \overline{\calB}_k 
    = \sum_{l=1}^k\diag(\bar{b}_{kl}(i))_{i\in\ZZ}\Lambda^{-i}  
\]
defined as 
\[
  \calB_k 
  = \left(\calW\Lambda^k\calW^{-1}\right)_{\ge 0},~~
  \overline{\calB}_k 
  = \left(\overline{\calW}\Lambda^{-k}
     \overline{\calW}^{-1}\right)_{<0}. 
\]
$(~~)_{\ge 0}$ and $(~~)_{<0}$ stand for the projection 
to the upper/lower-triangular part 
\[
  \left(\sum_{n=-\infty}^\infty a_n\Lambda^n\right)_{\ge 0} 
  = \sum_{n \ge 0}a_n\Lambda^n,~~
  \left(\sum_{n=-\infty}^\infty a_n\Lambda^n\right)_{<0} 
  = \sum_{n<0}a_n\Lambda^n. 
\]

A prototype of the $\eta$-matrix for this case 
is the $\ZZ\times(\ZZ\sqcup\ZZ)$ matrix 
\[
  \tilde{\eta}^{(2)}
  = \left(\begin{array}{c|c}
       \calW & \overline{\calW}
  \end{array}\right)
  = \left(\begin{array}{c|c}
    w_{i-j}(i) ~(i,j\in\ZZ) & \bar{w}_{i-j}~(i,j\in\ZZ)
    \end{array}\right). 
\]
The Sato-Wilson equations (\ref{B:SW-eq}) 
for $\calW,\overline{\calW}$ can be thereby 
cast into the matrix equations 
\beq
\begin{aligned}
  \frac{\rd\tilde{\eta}^{(2)}}{\rd t_k} 
  &= \calB_k\tilde{\eta}^{(2)} 
    - \tilde{\eta}^{(2)}\diag(\Lambda^k,0),\\
  \frac{\rd\tilde{\eta}^{(2)}}{\rd\bar{t}_k} 
  &= \overline{\calB}_k\tilde{\eta}^{(2)} 
    - \tilde{\eta}^{(2)}\diag(0,\Lambda^{-k}) 
\end{aligned}
  \label{B:etatil-tk-eq}
\eeq
for $k = 1,2,\ldots$.  

Given any value $s \in \ZZ$, 
we divide $\tilde{\eta}^{(2)}$ into eight blocks as 
\[
  \tilde{\eta}^{(2)} 
  = \left(\begin{array}{c|c|c|c}
    * & 0 & * & * \\\hline
    * & * & 0 & * 
    \end{array}\right).  
\]
The row index $i$ ranges over $i < s$ 
in the upper blocks and $i \ge s$ 
in the lower blocks, and the column index $j$ 
over $j < s$, $j \ge s$, $j < s$ and $j \ge s$, 
respectively, in the four blocks arranged 
from left to right. 
The third block in the upper half 
of $\tilde{\eta}^{(2)}$ is an upper-triangular matrix 
of the form 
\[
  h_1(s) = 
  \begin{pmatrix}
  \ddots & \ddots & \ddots & \vdots\\
  \cdots & 0  & \bar{w}_0(s-2) & \bar{w}_1(s-2)\\
  \cdots & \cdots & 0 & \bar{w}_0(s-1) 
  \end{pmatrix}, 
\]
and the second block in the lower half 
of $\tilde{\eta}^{(2)}$is a lower-triangular matrix 
of the form 
\[
  h_2(s) = 
  \begin{pmatrix}
  1 & 0 & \cdots & \cdots \\
  w_1(s+1) & 1  & 0 & \cdots\\
  w_2(s+2) & w_1(s+2) & 1 & \ddots \\
  \vdots & \ddots & \ddots &\ddots 
  \end{pmatrix}. 
\]
Using these matrices as building blocks, 
we construct the $\ZZ\times\ZZ$ matrix 
\[
  h(s) = \left(\begin{array}{c|c}
      h_1(s) & 0 \\\hline
      0 & h_2(s)
      \end{array}\right). 
\]
Since $h_1(s)$ and $h_2(s)$ are 
invertible triangular matrices, 
we can use $h(s)^{-1}$ to transform 
$\tilde{\eta}^{(2)}$ into 
\[
  \eta^{(2)}(s) = h(s)^{-1}\tilde{\eta}^{(2)},
\]
which becomes a matrix of the form 
\beq
  \eta(s)^{(2)} 
  = \left(\begin{array}{c|c|c|c}
    * & 0 & E & * \\\hline
    * & E & 0 & * 
    \end{array}\right).  
  \label{B:eta2(s)-blocks}
\eeq
This matrix represents a point of the two-component 
Sato Grassmannian 
\[
  \Gr^{(2)} \simeq \GL(\ZZ)\backslash\Fr(\ZZ,\ZZ\sqcup\ZZ). 
\]
Note that this identification is understood 
to depend on $s$ via the set partition 
of $\ZZ$ into $\ZZ_{<s}$ and $\ZZ_{\ge s}$.  
The matrix elements of the blocks other than 
$0$ and $E$ in $\eta(s)^{(2)}$ may be thought of 
as affine coordinates of a particular open cell 
of $\Gr^{(2)}$.  The evolution equations 
(\ref{B:etatil-tk-eq}) turn into 
the evolution equations 
\[
\begin{aligned}
  \frac{\rd\eta(s)^{(2)}}{\rd t_k}
  &= \calC_k(s)\eta(s)^{(2)} 
     - \eta(s)^{(2)}\diag(\Lambda^k,0),\\
  \frac{\rd\eta(s)^{(2)}}{\rd\bar{t}_k} 
  &= \overline{\calC}_k(s)\eta(s)^{(2)} 
     - \eta(s)^{(2)}\diag(0,\Lambda^{-k})
\end{aligned}
\]
for $\eta(s)^{(2)}$.  As is always the case 
for equations of this type, the coefficient matrices 
$\calC_k(s),\overline{\calC}_k(s)$ are uniquely 
determined by the evolution equations themselves. 

The $\eta$-matrices $\eta(s)^{(2)}$ enable us 
to translate the 2D Toda hierarchy to the language 
of pseudo-differential operators as well. 
This is a place where a $\calD$-module structure 
shows up \cite{Takasaki90LMP}.

\section{ASDYM hierarchy and auxiliary linear equations}

Let $\bsx = (x_0,x_1,\ldots)$ and 
$\bsy = (y_0,y_1,\ldots)$ be the space-time 
variables of the ASDYM hierarchy
\cite{Nakamura88,Takasaki90CMP,ACT93} 
and $\rd_{x_k}$ and $\rd_{y_k}$, $k = 1,2,\ldots$, 
the derivative operators 
\[
  \rd_{x_k} = \rd/\rd x_k,~~
  \rd_{y_k} = \rd/\rd y_k. 
\]
In the Lax formalism, the hierarchy comprises 
the zero-curvature equations 
\beq
\begin{aligned}
  [\rd_{x_j} - z\rd_{x_{j-1}} + A_j,\;
   \rd_{x_k} - z\rd_{x_{k-1}} + A_k] &= 0, \\
  [\rd_{y_j} - z\rd_{y_{j-1}} + B_j,\;
   \rd_{y_k} - z\rd_{y_{k-1}} + B_k] &= 0, \\
  [\rd_{x_j} - z\rd_{x_{j-1}} + A_j,\;
   \rd_{y_k} - z\rd_{y_{k-1}} + B_k] &= 0 
\end{aligned}
  \label{C:ZC-eq}
\eeq
for $j,k = 1,2,\ldots$.  $A_k$'s and $B_k$'s 
are $r \times r$ matrix-valued gauge potentials. 
Expanded in powers of $z$, these equations 
give an infinite tower of differential equations 
for these gauge potentials. 

Introducing a matrix-valued potential, 
one can convert those equations into equations 
for that potential.  In fact, there are 
two choices of such a potential.  
Note that the first and second equations 
of (\ref{C:ZC-eq}) split into the equations 
\[
  [\rd_{x_j} + A_j,\; \rd_{x_k} + A_k] = 0, ~~
  [\rd_{y_j} + B_j,\; \rd_{y_k} + B_k] = 0
\]
and 
\[
  \rd_{x_{j-1}}A_k - \rd_{x_{k-1}}A_j = 0, ~~
  \rd_{y_{j-1}}B_k - \rd_{y_{k-1}}B_j = 0. 
\]
They ensure the existence of $r \times r$ 
matrix-valued potentials $J,K$ such that 
\beq
  A_k = - \rd_{x_k}J\cdot J^{-1} = \rd_{x_{k-1}}K,~~
  B_k = - \rd_{y_k}J\cdot J^{-1} = \rd_{y_{k-1}}K. 
  \label{C:AB-JK-rel}
\eeq
On the other hand, the third equation of 
(\ref{C:ZC-eq}) gives 
\[
  [\rd_{x_j} + A_j,\;\rd_{y_k} + B_k] = 0 
\]
and 
\[
  \rd_{x_{j-1}}B_k - \rd_{y_{j-1}}A_k = 0. 
\]
Substituting the expression (\ref{C:AB-JK-rel}) 
into $A_k$'s and $B_k$'s in these equations, 
one obtains the differential equations 
\beq
  \rd_{x_{j-1}}(\rd_{y_k}J\cdot J^{-1}) 
  - \rd_{y_{j-1}}(\rd_{x_k}J\cdot J^{-1}) = 0 
  \label{C:Yang-eq}
\eeq
for $J$ and 
\beq
  \rd_{x_j}\rd_{y_{k-1}}K - \rd_{y_k}\rd_{x_{j-1}}K 
  + [\rd_{x_{j-1}}K,\rd_{y_{k-1}}K] = 0 
  \label{C:NCZ-eq}
\eeq
for $K$.  These equations are generalizations 
of the Yang equation \cite{Yang77} and 
the Leznov-Mukhtarov-Parkes-Chalmers-Siegel equation 
\cite{LM87,Parkes92,CZ96} in 4D space-time. 
In particular, by letting $j = k = 1$ and defining 
the 4D (possibly complexified) 
space-time coordinates $x,y,\tilde{x},\tilde{y}$ as 
\[
  x = x_0,~~ y = y_0,~~ 
  \tilde{x} = y_1,~~ \tilde{y} = x_1, 
\]
(\ref{C:NCZ-eq}) and (\ref{C:Yang-eq}), 
respectively, become the Yang equation 
\[
  \rd_x\left(\rd_{\tilde{x}}J\cdot J^{-1}\right) 
  - \rd_y\left(\rd_{\tilde y}J\cdot J^{-1}\right) = 0 
\]
and the Leznov-Mukhtarov-Parkes-Chalmers-Siegel equation 
\[
  \rd_{\tilde{y}}\rd_yK - \rd_{\tilde{x}}\rd_xK 
  + [\rd_xK,\rd_yK] = 0.  
\]

The zero curvature equations (\ref{C:ZC-eq}) 
are associated with the auxiliary linear equations 
\beq
  (\rd_{x_k} - z\rd_{x_{k-1}} + A_k)\Psi(z) = 0,~~
  (\rd_{y_k} - z\rd_{y_{k-1}} + B_k)\Psi(z) = 0. 
  \label{C:Psi-xy-eq}
\eeq
One can consider two solutions $W(z),\overline{W}(z)$ 
with Laurent expansion of different types: 
\[
  W(z) = E_r + \sum_{n=1}^\infty w_nz^{-n},~~
  \overline{W}(z) = \sum_{n=0}^\infty\bar{w}_nz^n. 
\]
The auxiliary linear equations (\ref{C:Psi-xy-eq}) 
turn into differential equations 
\beq
\begin{aligned}
  \rd_{x_k}w_n - \rd_{x_{k-1}}w_{n+1} + A_kw_n &= 0,\\
  \rd_{y_k}w_n - \rd_{y_{k-1}}w_{n+1} + B_kw_n &= 0
\end{aligned}
  \label{C:wn-xy-eq}
\eeq
for $w_n$'s and 
\[
\begin{aligned}
  \rd_{x_k}\bar{w}_n - \rd_{x_{k-1}}\bar{w}_{n-1} 
    + A_k\bar{w}_n &= 0,\\
  \rd_{y_k}\bar{w}_n - \rd_{y_{k-1}}\bar{w}_{n-1} 
    + B_k\bar{w}_n &= 0
\end{aligned}
\]
for $\bar{w}_n$'s.  In particular, we have 
\[
  - \rd_{x_{k-1}}w_1 + A_k = 0, ~~
  - \rd_{y_{k-1}}w_1 + B_k = 0 
\]
and 
\[
  \rd_{x_k}\bar{w}_0 + A_k\bar{w}_0 = 0, ~~
  \rd_{y_k}\bar{w}_0 + B_k\bar{w}_0 = 0, 
\]
hence 
\[
\begin{aligned}
  A_k &= - \rd_{x_k}\bar{w}_0\cdot\bar{w}_0^{-1} 
      = \rd_{x_{k-1}}w_1,\\
  B_k &= - \rd_{y_k}\bar{w}_0\cdot\bar{w}_0^{-1} 
      = \rd_{y_{k-1}}w_1. 
\end{aligned}
\]
In view of (\ref{C:AB-JK-rel}), 
this means that $\bar{w}_0$ and $w_1$ 
can be identified with $K$ and $J$:
\[
  J = \bar{w}_0,~~
  K =w_1.
\]

\end{document}